\documentclass[journal,twoside,web]{ieeecolor}

\let\labelindent\relax

\usepackage{generic}
\usepackage{cite}
\usepackage{amsmath,amssymb,amsfonts,amsthm}
\usepackage{graphicx}
\usepackage{algorithm,algorithmic}
\usepackage{hyperref}
\usepackage{textcomp}
\usepackage{enumitem}
\usepackage{verbatim}
\usepackage{xcolor}
\usepackage{textcomp}
\usepackage{bm}
\usepackage{array}
\hypersetup{hidelinks}
\def\BibTeX{{\rm B\kern-.05em{\sc i\kern-.025em b}\kern-.08em
    T\kern-.1667em\lower.7ex\hbox{E}\kern-.125emX}}
\theoremstyle{plain}
\newtheorem{theorem}{Theorem}
\newtheorem{corollary}{Corollary}

\newtheorem{definition}{Definition}

\newcommand{\R}{\mathbb{R}}
\newcommand{\G}{\mathbf{G}}
\newcommand{\C}{\mathbf{C}}
\newcommand{\RR}{\mathbf{R}}
\newcommand{\mysqrt}[1]{\left( #1 \right)^\frac{1}{2}}
\newcommand{\lcurl}[1]{\left\{ #1 \right\}}
\newcommand{\lpar}[1]{\left( #1 \right)}

\newcommand{\dnu}[1]{\frac{d}{d \nu} #1}
\newcommand{\suchthat}{\,\ifnum\currentgrouptype=16 \middle\fi|\,}
\newcommand{\pq}{\phi_1}
\newcommand{\pw}{\phi_2}
\newcommand{\ophi}{\overline{\phi}}
\newcommand{\ophione}{\overline{\phi}_1}
\newcommand{\ophitwo}{\overline{\phi}_2}
\newcommand{\oxone}{\overline{X}_1}
\newcommand{\oxtwo}{\overline{X}_2}
\newcommand{\OG}{\overline{G}}
\newcommand{\xai}[1]{\overline{X}_{A, #1}}
\DeclareMathOperator*{\argmax}{arg\,max}

\definecolor{myorange}{HTML}{FF9300}
\def\orangelozenge{\mathbin{\color{myorange}\blacklozenge}}
\newcolumntype{L}{>{\centering\arraybackslash}m{0.45\linewidth}}

\begin{document}
\title{Alliance Mechanisms in General Lotto Games}
\author{Vade Shah, \IEEEmembership{Graduate Student Member, IEEE}, Jason R. Marden, \IEEEmembership{Fellow, IEEE}
\thanks{This work is supported by AFOSR grants \#FA9550-21-1-0203 and \#FA9550-25-1-0245 and the NSF GRFP grant \#2139319.}%
\thanks{V. Shah ({\tt\small vade@ucsb.edu}) and J. R. Marden are with the Department of Electrical and Computer Engineering at the University of California, Santa Barbara, CA.}%
}

\maketitle

\begin{abstract}
    How do different alliance mechanisms compare? In this work, we analyze various methods of forming an alliance in the Coalitional General Lotto game, a simple model of competitive resource allocation. In the game, Players 1 and 2 independently compete against a common Adversary by allocating their limited resource budgets towards separate sets of contests; an agent wins a contest by allocating more resources towards it than their opponent. In this setting, we study three alliance mechanisms: budget transfers (resource donation), contest transfers (contest redistribution), and joint transfers (both simultaneously). For all three mechanisms, we study when they present opportunities for collective improvement (the sum of the Players' payoffs increases) or mutual improvement (both Players' individual payoffs increase). In our first result, we show that all three are fundamentally different with regards to mutual improvement; in particular, mutually beneficial budget and contest transfers exist in distinct, limited subsets of games, whereas mutually beneficial joint transfers exist in almost all games. However, in our second result, we demonstrate that all three mechanisms are equivalent when it comes to collective improvement; that is, collectively beneficial budget, contest, and joint transfers exist in almost all game instances, and all three mechanisms achieve the same maximum collective payoff. Together, these results demonstrate that differences between mechanisms depend fundamentally on the objective of the alliance.
\end{abstract}

\begin{IEEEkeywords}
game theory, multi-agent systems, coalition formation, Colonel Blotto, General Lotto
\end{IEEEkeywords}

\section{Introduction}\label{sec:intro}

\IEEEPARstart{F}{orming} an alliance is a fundamental strategy for gaining competitive advantage in adversarial multi-agent settings such as product development \cite{elmuti2001overview, culpan2002global, king2003complementary}, energy markets \cite{van2015power, de2020cooperatives, tushar2018transforming}, and political elections \cite{blais2007making, di1998electoral, spoon2015alone}. In the control community, the questions of whether and when to form alliances have long been studied in the context of multi-agent coordination \cite{bitar2012, kaewpuang2013framework, zhang2019competition, zhang2020} and network security \cite{saad2009coalitional, bou2013, tushar2018transforming}, where coalitional and cooperative game theory have been applied to characterize or compute stable alliances and quantify the benefits of alliance membership. Despite extensive work on the stability and benefits of alliances, the actual mechanisms by which alliances are formed—such as resource exchange, information sharing, or treaties and truces—have received less attention. While these mechanisms are implicitly assumed in existing analyses, their unique influences on agent performance remain underexplored, warranting further study.

To investigate how different mechanisms compare, this work turns to the \emph{Coalitional General Lotto} game, a well-studied adversarial resource allocation problem. In the game (Figure \ref{fig:simple_example}), Player 1 and Player 2 independently compete against a common Adversary by allocating their limited resource \emph{budgets} (e.g., capital, energy) across multiple valued \emph{contests} (e.g., consumer markets, territories). An agent wins a contest and its associated value by allocating more resources towards it than their opponent. First introduced in \cite{kovenock2012coalitional}, the model offers a straightforward testbed for analyzing alliances, since Players 1 and 2 (entangled in a so-called `enemy-of-my-enemy' relationship) may collude against the Adversary to improve their individual outcomes. Over the past decade, a growing body of work has extended this model to study different alliance mechanisms in the Coalitional General Lotto game \cite{shah2024battlefield, shah2024inefficient, chandan2020showing, chandan2022art, gupta2014three, gupta2014three2, heyman2018colonel} and other similar variants \cite{diaz2023beyond, diaz2024strategic, paarporn2021division}.

\begin{figure}
    \centering
    \includegraphics[width=\linewidth]{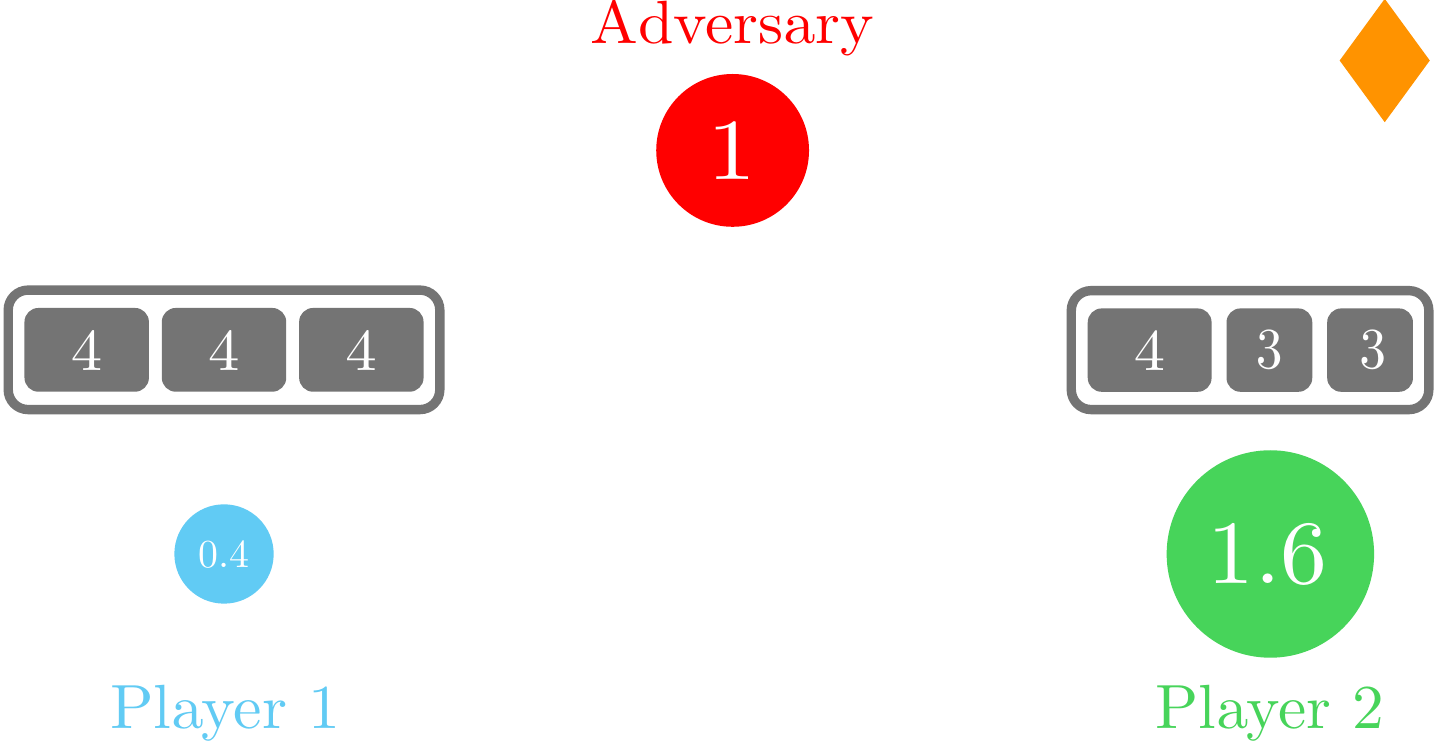}
    \caption{An example of a Coalitional General Lotto Game. Players 1 and 2 are equipped with budgets $X_1 = 0.4$ and $X_2 = 1.6$, respectively, and compete against the Adversary across sets of contests with cumulative valuation $4 + 4 + 4 = 12$ and $4 + 3 +  3 = 10$; we summarize these game parameters with the tuple $(12, 10, 0.4, 1.6)$. Throughout the text, we refer to this example using the symbol $\orangelozenge$.}
    \label{fig:simple_example}
\end{figure}

Among the mechanisms studied in the literature, perhaps the most natural is the \emph{budget transfer} introduced in \cite{kovenock2012coalitional}. In a budget transfer, one Player (say, Player 1) donates a portion of their budget to the other Player, so Player 1's budget decreases while Player 2's increases. Two key results emerge in this setting. First, in a nontrivial subset of game instances, the Players can form a \emph{mutually beneficial} alliance, one that increases both Players’ individual payoffs, through a budget transfer—even though one Player appears to simply give away resources. The intuition is that the transfer endows the recipient with greater resources, but it also shifts the Adversary’s focus away from the donor, thereby improving both Players' outcomes. The second result is that in almost all game instances, the Players can form a \emph{collectively beneficial} alliance, one that increases the collective payoff (the sum of their payoffs), through a budget transfer.

In this work, we introduce a novel alliance mechanism as an alternative to budget transfers, which we refer to as a \emph{contest transfer}. In a contest transfer, one Player (say, Player 1) forgoes competing in a subset of contests, allowing Player 2 to compete for them instead. Thus, the total value of contests available to Player 1 decreases while that for Player 2 increases. The notion of a contest transfer is well-motivated by a number of real-world phenomena such as electoral endorsements in political campaigns \cite{spoon2015alone} and non-compete agreements in business domains \cite{igami2022measuring}. However, despite their prevalence, the strategic opportunities offered by contest transfers have not been studied in a game-theoretic framework.

The goal of this work is to compare budget transfers, contest transfers, and \emph{joint transfers}, i.e., simultaneous transfers of both budgets and contests, with respect to their ability to generate different kinds of alliances. Specifically, we ask:

\vspace{0.2cm}

\begin{enumerate}[label=Q\arabic*:]
    \item Do budget, contest, and joint transfers differ in their ability to generate \emph{mutually} beneficial alliances?
    \item Do budget, contest, and joint transfers differ in their ability to generate \emph{collectively} beneficial alliances?
\end{enumerate}

We address these questions by asking whether a set of game parameters—the Players' budgets and the cumulative contest valuations (Figure \ref{fig:simple_example})—generates opportunities for mutually or collectively beneficial transfers of the given form. In the main text, we answer each question as follows:

\vspace{0.2cm}

\noindent A1: \, In \textbf{Theorem \ref{thm:mut_ben_diff}}, we show that when it comes to mutual benefit, the three mechanisms are distinct: while budget and contest transfers provide opportunities for mutual improvement only in limited, partially overlapping subsets of the parameter space, \emph{joint} transfers do so almost everywhere (in a measure-theoretic sense).

\vspace{0.2cm}

\noindent A2: \, In \textbf{Theorem \ref{thm:col_ben_equiv}}, we demonstrate that when Players seek to form collectively beneficial alliances, budget, contest, and joint transfers are equivalent: they all arise in the same subset of the parameter space, and they all yield the same maximum collective payoff\footnote{Comparing the collective payoff induced by different kinds of transfers is reasonable when studying collective benefit (Q2/A2), but this comparison is not particularly well-defined for mutual benefit (Q1/A1), since improvements in collective payoff do not imply mutual improvements in individual payoffs.}.

\vspace{0.2cm}

Taken together, our results reveal that the choice of the alliance mechanism should depend on the alliance objective: when Players seek to maximize collective payoffs, all three mechanisms are equivalent and their choice is irrelevant, but when they seek mutual improvement, the mechanisms diverge and their choice is crucial. Moreover, joint transfers provide an especially powerful tool, enabling mutual gains that neither budget nor contest transfers achieve alone.

\section{Model}

In this section, we introduce the model of the Coalitional General Lotto game in which two Players independently compete against a common Adversary across disjoint sets of contests. For ease of exposition, we begin with a discussion of the classical one-versus-one General Lotto game between a single Player and the Adversary.

\subsection{General Lotto Game}

In the \emph{General Lotto} game, two agents (say, Player 1 and the Adversary) compete across a set of $n$ contests. Player 1 and the Adversary are endowed with arbitrarily divisible resource budgets $X_1 \in \R_{>0}$ and $X_A \in \R_{>0}$, respectively, which they allocate across the contests. The agent who allocates a greater level of resources towards the $k^\mathrm{th}$ contest wins its valuation $v^k \in \R_{>0}$. The budgets and the contest valuations are known to both agents, but neither knows how their opponent will allocate their budget across the contests. An allocation decision for Player 1 is a vector $D_1 \in \R_{\geq 0}^n$, where $D_1^k$ is the amount of budget that they allocate to contest $k$; the decision for the Adversary $D_A$ is defined similarly. Player 1's payoff for a given pair $(D_1, D_A)$ is of the form
\begin{equation*}
    U_1^{\rm GL}(D_1, D_A) = \sum_{k=1}^n v^k \cdot I \{D_{1}^k \geq D_{A}^k\},
\end{equation*}
where $I\{\cdot\}$ is the usual indicator function, and the Adversary's payoff is $U_{A}^{\rm GL}(D_{1}, D_{A}) = \sum_{k=1}^n v^k - U_{1}^{\rm GL}(D_1, D_A)$.

A \emph{strategy} for Player 1 is a $n$-variate distribution function $F_1$ satisfying
\begin{equation}
    \mathbb{E}_{D_1 \sim F_1}\left[\sum_{k = 1}^n D_1^k \right] \leq X_1,
\end{equation}
meaning that Player 1 meets their budget constraint \emph{in expectation}. The Adversary's strategy $F_A$ is defined similarly.
The agents' Nash equilibrium strategies and payoffs in this setting are well-understood and have been characterized in \cite{kovenock2021generalizations}. For ease of exposition, we omit the equilibrium strategies and instead present only the equilibrium payoffs in the following Corollary from \cite{kovenock2021generalizations}:
\begin{corollary}[Kovenock and Roberson, 2021]
    For any General Lotto game with budgets $X_1$, $X_A$ and contest valuations $v^1, \dots, v^n$, a Nash equilibrium exists, and the unique Nash equilibrium payoffs for Player 1 and the Adversary are
    \begin{align}
        U_1^{\rm NE}(\phi, X_1, X_A) &\triangleq \begin{cases}
            \phi \left( \frac{X_1}{2 X_A} \right) & X_1 \leq X_A \\
            \phi \left( 1 - \frac{X_A}{2 X_1} \right) & X_1 > X_A \label{eq:player_payoff}
        \end{cases} \\
        \text{and } \; U_A^{\rm NE}(\phi, X_1, X_A) &\triangleq \phi - U_1^{\rm NE}(\phi, X_1, X_A), \label{eq:adversary_payoff}
    \end{align}
    respectively, where $\phi \triangleq \sum_{k = 1}^n v^k$.
\end{corollary}

\begin{figure*}
    \includegraphics[width=\textwidth]{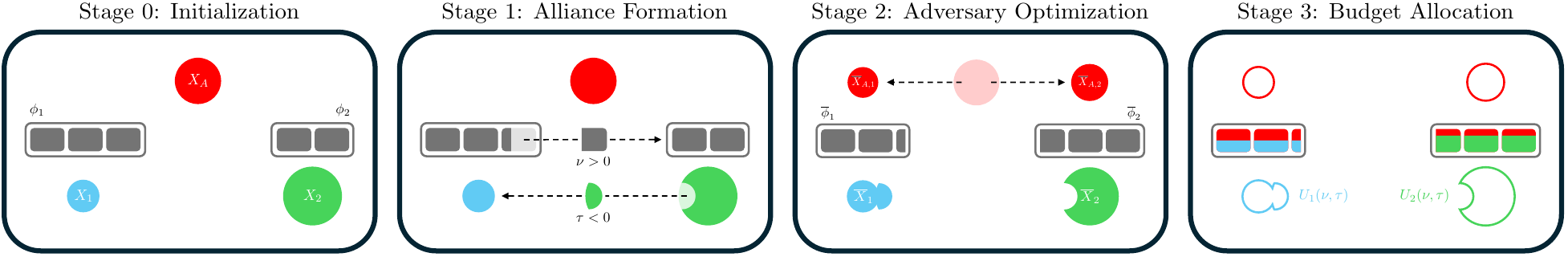}
    \caption{The stages of the Coalitional General Lotto game. In Stage 0, the game is initialized. In Stage 1, the two Players may form an alliance through a budget transfer, contest transfer, or both. In Stage 2, the Adversary determines how to optimally split their budget between the two standard General Lotto games. In Stage 3, the agents allocate their budgets and receive their payoffs.}
    \label{fig:stages}
\end{figure*}

\subsection{Coalitional General Lotto Game}

To study alliance mechanisms, we analyze the \emph{Coalitional General Lotto game}, a variant in which Players 1 and 2 compete against a common Adversary in separate standard General Lotto games. Our focus is on opportunities for the Players to form alliances, where informally, an alliance consists of transferring contests, resources, or both from one game to the other. Formally, the interaction unfolds over multiple stages, as illustrated in Figure \ref{fig:stages} and described in detail below.

\vspace{.2cm}
\noindent\textbf{Stage 0 (Initialization):} The game is initialized. Player 1, Player 2, and the Adversary have budgets $X_1$, $X_2$, and $X_A$, respectively, where their budgets are normalized such that $X_A = 1$. Player $i \in \{1, 2\}$ competes against the Adversary in General Lotto game $i$, where the $n_i$ contests have cumulative valuation $\phi_i$. A Coalitional General Lotto game instance is thus fully parameterized by the tuple $G = (\pq, \pw, X_1, X_2) \in \G = \R_{>0}^4$ which all agents observe before Stage 1.

\vspace{.2cm}
\noindent\textbf{Stage 1 (Alliance Formation):} The Players may form an alliance through a transfer of contests, budgets, or both.
\begin{definition}[Types of transfers]
    \ \\ \vspace{-1em}
    \begin{enumerate}[label=\alph*)]
        
        \item A \emph{budget transfer} $$\tau \in (-X_2, X_1)$$ is the amount of budget transferred from Player 1 to Player 2; a negative value of $\tau$ indicates that the net amount is transferred from Player 2 to Player 1.

        \item A \emph{contest transfer} $$\nu \in (-\pw, \pq)$$ is the amount of contest valuation transferred from Player 1 to Player 2; a negative value of $\nu$ indicates that the net amount is transferred from Player 2 to Player 1.
        
        \item A \emph{joint transfer} $$(\tau, \nu) \in (-X_2, X_1) \, \times \, (-\pw, \pq)$$ is a pair specifying the amount of contest valuation and budget transferred from Player 1 to Player 2.
    \end{enumerate}
\end{definition}

\noindent Transfers effectively change the parameters of a Coalitional General Lotto game. Specifically, the post-transfer budgets of Players 1 and 2 are
\begin{equation*}
    \oxone \triangleq X_1 - \tau \qquad \text{ and } \qquad \oxtwo \triangleq X_2 + \tau,
\end{equation*}
and the post-transfer contest valuations in games 1 and 2 are
\begin{equation*}
    \ophione \triangleq \pq - \nu \qquad \text{ and } \qquad \ophitwo \triangleq \pw + \nu,
\end{equation*}
respectively; we denote the new parameters induced by transfers $(\tau, \nu)$ as
$$\OG \triangleq (\ophione, \ophitwo, \oxone, \oxtwo),$$ which all agents (including the Adversary) observe before proceeding to Stage 2.\footnote{We take the budgets and contests to be arbitrarily divisible, meaning that any amount of contest valuation or budget can be feasibly transferred. Here, we remark that the decision regarding whether to perform any kind of transfer, and if so, how much to transfer, is one that the Players make jointly. Our focus is on when transfer parameters generate different kinds of alliances, not how the Players choose and agree upon these parameters.} Here, note that we write $\overline{G}$ with no dependence on the transfer $(\tau, \nu)$ for notational simplicity.

\vspace{.2cm}
\noindent\textbf{Stage 2 (Adversary Optimization):} The Adversary performs a best response to any transfers between the Players by optimally dividing their budget between the two General Lotto games. Mathematically, this entails solving the problem
\begin{equation}\label{eq:adversary_optimization}
    \argmax_{\underset{X_{A,1} + X_{A,2} \leq 1}{X_{A,1}, X_{A,2} \geq 0}} U_A^{\rm NE}(\ophione, \oxone, X_{A,1}) + U_A^{\rm NE}(\ophitwo, \oxtwo, X_{A,2}).
\end{equation}
We denote the solution to this problem as $\xai{1}$ and $\xai{2}$, which represents the Adversary's optimal allocations to Lotto games 1 and 2, respectively. The solution to \eqref{eq:adversary_optimization} is derived in \cite{kovenock2012coalitional}, which we summarize in Appendix \ref{app:cases}. 

\vspace{.2cm}
\noindent\textbf{Stage 3 (Budget Allocation):} In the third and final stage, the Players and the Adversary allocate their budgets towards the contests according to their equilibrium General Lotto strategies and receive their corresponding equilibrium payoffs \eqref{eq:player_payoff}, \eqref{eq:adversary_payoff}. Here, we rewrite the agents' payoffs to highlight the dependence on the transfer $(\tau, \nu)$; in particular, we write
\begin{align}
    U_1(\tau, \nu \suchthat G) &\triangleq U_1^{\rm NE}(\ophione, \oxone, \xai{1}) \label{eq:player1} \\
    U_2(\tau, \nu \suchthat G) &\triangleq U_1^{\rm NE}(\ophitwo, \oxtwo, \xai{2}) \label{eq:player2}
\end{align}
for Players 1 and 2, respectively. When clear, we may omit the dependence on $G$ and simply write $U_1(\tau, \nu)$ or $U_2(\tau, \nu)$.

Having established this framework, we now investigate how budget, contest, and joint transfers compare in creating opportunities for beneficial alliances.

\section{Results}

In this section, we answer the two main questions Q1 and Q2 introduced in Section \ref{sec:intro}, comparing different transfer mechanisms with respect to their ability to generate alliances. We begin by studying mutually beneficial transfers.

\subsection{Mutually Beneficial Alliances}

We first focus on alliances formed through mutually beneficial transfers, which we define as follows:

\begin{figure*}
    \includegraphics[width=\textwidth]{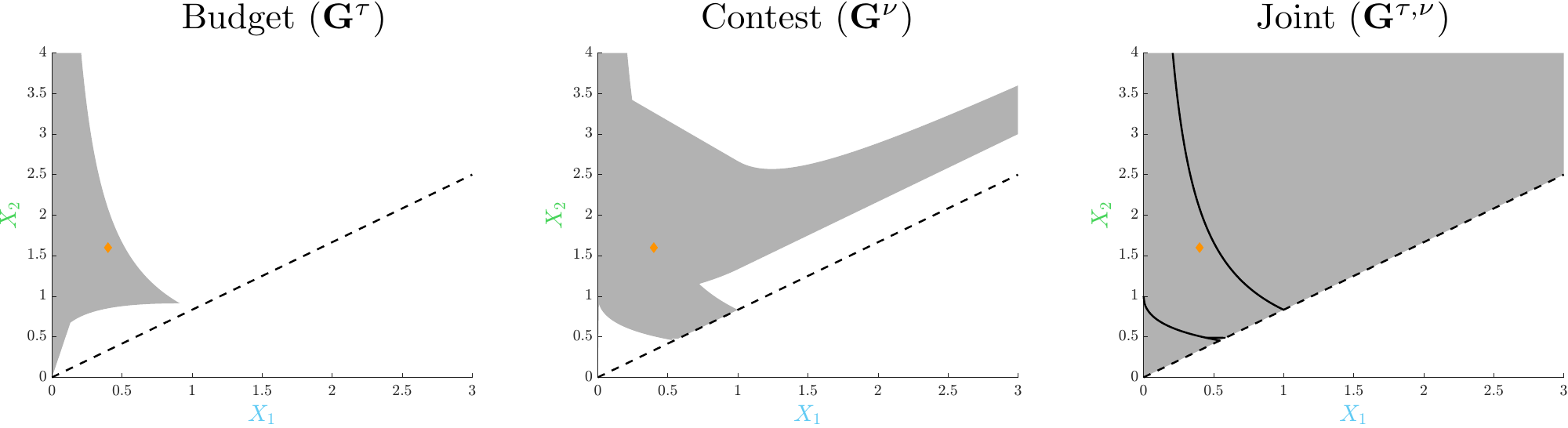}
    \caption{The regions in which mutually beneficial budget (left), contest (center), and joint transfers (right) exist, shown as grey shaded regions, plotted in the $X_1$-$X_2$ space for fixed $\phi_1 = 12$, $\phi_2 = 10$. Only games where $X_1 / \phi_1 \leq X_2 / \phi_2$ are plotted to avoid redundancy. The lines in the rightmost plot depict the measure-zero subset $\G \setminus \G^{\tau, \nu}$, and the game depicted in Figure \ref{fig:simple_example} is indicated using $\orangelozenge$. Observe that there are some games (e.g., $\orangelozenge$) that belong to both $\G^\tau$ and $\G^\nu$, but there are also games for which there exist mutually beneficial contest transfers, but not budget transfers (e.g., any point in $\G^\nu$ with $X_1 > 1$ in the center plot), and vice versa.}
    \label{fig:mutually_beneficial_regions}
\end{figure*}

\begin{definition}[Mutually beneficial transfers]
\ \\ \vspace{-1em}

    A transfer $(\tau', \nu')$ is \emph{mutually beneficial} if
    \begin{equation}\label{eq:mut_ben_transfer}
        U_1(\tau', \nu') > U_1(0, 0) 
        \; \, \text{and} \; \, U_2(\tau', \nu') > U_2(0, 0).
    \end{equation}
    
\end{definition}

Put simply, a transfer is mutually beneficial if it increases both Players' individual payoffs. Now, recall our first main question (Q1): do budget, contest, and joint transfers differ in their ability to generate \emph{mutually} beneficial alliances? In our first result, we characterize the subsets of games in which mutually beneficial budget, contest, and joint transfers arise and show that these subsets are structurally different, indicating that the three mechanisms do indeed differ.

\begin{theorem}[Properties of $\G^\tau$, $\G^\nu$, and $\G^{\tau, \nu}$]\label{thm:mut_ben_diff}
Let $\G^\tau$, $\G^\nu$, and $\G^{\tau, \nu}$ denote the subsets of $\G$ in which mutually beneficial budget, contest, and joint transfers exist, respectively.
\begin{enumerate}[label=\alph*)]
    \item $\G^\tau$, $\G^\nu$, and $\G^{\tau, \nu}$ are nonempty.
    \item $\G^\tau \cap \G^\nu$ is nonempty.
    \item $\G^\tau \setminus \G^\nu$, $\G^\nu \setminus \G^\tau$, and $\G^{\tau, \nu} \setminus (\G^\tau \cup \G^\nu)$ are nonempty.
    \item $\G \setminus \G^{\tau, \nu}$ has measure zero.
\end{enumerate}
\end{theorem}

\begin{proof}
    See Appendix \ref{app:contest_transfer_proof}.
\end{proof}

Theorem \ref{thm:mut_ben_diff} states that mutually beneficial budget and contest transfers arise in limited, non-overlapping subsets of the parameter space, and mutually beneficial joint transfers arise in a substantially larger region that contains both individual subsets. The proof of this result relies on a novel and nontrivial characterization of $\G^\nu$, the subset of games that generate mutually beneficial contest transfers. The characterization involves a careful, case-by-case analysis, which reflects the complexity inherent in mutual improvement. Statements a)–d) in Theorem \ref{thm:mut_ben_diff} provide a more refined understanding of the actual differences between the mechanisms. Below, we parse through and discuss each in greater detail.

\begin{enumerate}[label=\alph*),align=left, leftmargin=0pt,labelindent=1em,listparindent=0pt, itemindent=!]
    \item Mutually beneficial budget, contest, and joint transfers exist. The novel characterization of $\G^\nu$ in Appendix \ref{app:contest_transfer_proof} complements the well-established result regarding the existence of $\G^\tau$ from \cite{kovenock2012coalitional}. The existence of $\G^{\tau, \nu}$, characterized in a conference version of this paper \cite{shah2024battlefield}, follows trivially from the existence of $\G^\nu$ and $\G^\tau$, but its structure is fundamentally different. All three subsets are illustrated in Figure \ref{fig:mutually_beneficial_regions}.
    
    \item Mutually beneficial budget and contest transfers are not mutually exclusive; that is, there are instances where either mechanism can generate a mutually beneficial alliance. This can be seen in the left two panels of Figure \ref{fig:mutually_beneficial_regions}. 

    \item In certain game instances, budget transfers may generate mutually beneficial alliances, but not contest transfers, and vice versa, as shown in the left two panels of Figure \ref{fig:mutually_beneficial_regions}. Furthermore, there are instances where joint transfers generate mutually beneficial alliances, but contest or budget transfers do not. Thus, when Players seek to form mutually beneficial alliances, they may be limited in their choice of mechanism.

    \item In \emph{almost all} game instances—where `almost all' is used in the measure-theoretic sense—joint transfers can generate mutually beneficial alliances, as depicted in the rightmost panel of \ref{fig:mutually_beneficial_regions}. The proof proceeds by demonstrating that the gradients of the payoff functions $U_1$ and $U_2$ are diametrically opposed (i.e., $\nabla U_1 = -\gamma \nabla U_2$ for some $\gamma > 0$) or ill-defined only along a measure zero subset of $\G$, which is sufficient for establishing the result \cite{shah2024battlefield}. This highlights that joint transfers offer far greater opportunities than budget or contest transfers.
\end{enumerate}

Collectively, these results answer our second research question definitively: mechanism choice is critical for mutually beneficial alliances. While budget and contest transfers succeed only in limited, partially overlapping regions, joint transfers provide opportunities for mutual improvement in virtually all scenarios. This suggests that Players seeking individual gains should prioritize joint mechanisms when possible.

\begin{figure*}
    \includegraphics[width=\textwidth]{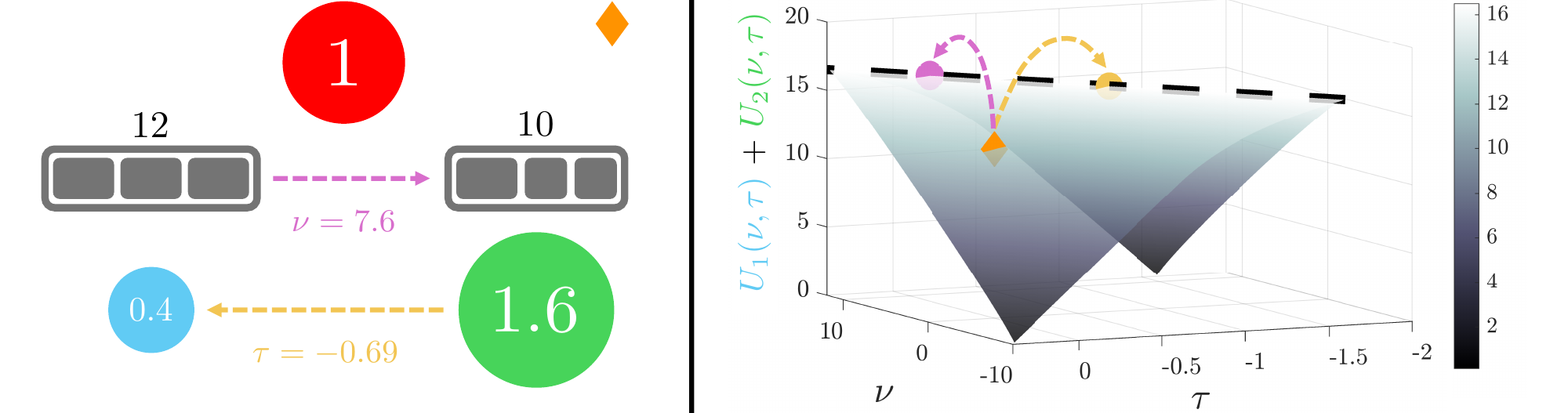}
    \caption{Left: an example of a Coalitional General Lotto game ($\orangelozenge$); Player $1$ and $2$ have budgets of $X_1 = 0.4$ and $X_2 = 1.6$ and compete for contests with valuations $12$ and $10$, respectively. The Players either perform a budget transfer $\tau = -0.69$ (mustard) or a contest transfer $\nu = 7.6$ (magenta). Right: the Players' collective payoff as a function of the transfer amount $\nu$, $\tau$. The symbol $\orangelozenge$ shows the collective payoff for the nominal setting with no transfers, i.e., $(\tau, \nu) = (0, 0)$. Observe that whether the Players perform a budget transfer $\tau = -0.69$ (mustard) or a contest transfer $\nu = 7.6$ (magenta), they achieve the same maximum collective payoff (dashed black line).}
    \label{fig:comparison_example}
\end{figure*}

\subsection{Collectively Beneficial Alliances}

Next, we consider alliances formed through a collectively beneficial transfer, which we define as follows:

\begin{definition}[Collectively beneficial transfers]
\ \\ \vspace{-1em}

    \item A transfer $(\tau^*, \nu^*)$ is \emph{collectively beneficial} if \begin{equation}\label{eq:joint_ben_transfer}
        U_1(\tau^*, \nu^*) + U_2(\tau^*, \nu^*) > U_1(0, 0) + U_2(0, 0).
    \end{equation}
    
\end{definition}

As the name suggests, collectively beneficial transfers improve the Players' collective payoff, but they need not be mutually beneficial. Now, recall our second main question (Q2): do budget, contest, and joint transfers differ in their ability to generate \emph{collectively} beneficial alliances? Our second result establishes that in this setting, the answer turns out to be yes—the three mechanisms are equivalent.

\begin{theorem}[Equivalence of transfer types]\label{thm:col_ben_equiv}
    For any Coalitional General Lotto game $G \in \G$,
    \begin{equation}\label{eq:collective_payoff_equivalence}
        \begin{aligned}
            \max_{\tau} \ U_1(\tau, 0) + U_2(\tau, 0)
            = \, &\max_{\tau, \nu} \ U_1(\tau, \nu) + U_2(\tau, \nu) \\
            = \, &\max_{\nu} \ U_1(0, \nu) + U_2(0, \nu).
        \end{aligned}
    \end{equation}
\end{theorem}

\begin{proof}
    See Appendix \ref{app:equivalence_proof}.
\end{proof}

Theorem \ref{thm:col_ben_equiv} highlights that when it comes to maximizing the collective payoff, there is effectively no difference between budget, contest, and joint transfers. Note that this Theorem says two things: first, opportunities for collective improvement arise in the same subsets of games for all three mechanisms, and second, that the magnitude of improvement is also the same. The key implication of this result is that performing budget and contest transfers simultaneously cannot enhance collective performance; forming an alliance through either mechanism alone is sufficient. This is exemplified in Figure \ref{fig:comparison_example}; no matter which transfer the Players perform, their maximum collective payoff is the same.

\section{Conclusion}

In this paper, we study three alliance mechanisms—resource sharing (budget transfers), non-compete agreements (contest transfers), and combinations of the two (joint transfers)—using the framework of the Coalitional General Lotto game, a competitive resource allocation problem where two players compete against a common Adversary. We compare these transfers on two different performance objectives, collective benefit and mutual benefit. While intuition might suggest that different transfer mechanisms would consistently yield different outcomes, our results demonstrate otherwise: although the mechanisms are equivalent when it comes to collective benefit (Theorem \ref{thm:col_ben_equiv}), they differ greatly in their ability to generate opportunities for mutual benefit (Theorem \ref{thm:mut_ben_diff}).

Returning to our opening question about comparing alliance mechanisms: the answer depends entirely on what the alliance seeks to achieve. Practically, this implies that organizations seeking to maximize joint performance can choose any alliance mechanism based on convenience or feasibility. However, when individual partners require guaranteed improvement, the choice of mechanism becomes crucial. In this setting, our results suggest that joint approaches should be the default consideration, as they offer far greater flexibility than single-mechanism alliances.

Several interesting directions emerge from this work. An important future study entails extending these results to the network setting with several competitors and contests; a similar study has been performed in \cite{diaz2023beyond} for a different payoff model. Another question is to analyze whether other strategic maneuvers, such as information sharing (if the Players hold priors on the Adversary's allocation strategy, for example), would also create distinct alliance opportunities.

\appendix

\subsection{Adversary's Best Response}\label{app:cases}

The Adversary's best response to any transfer performed by the Players can be completely characterized by seven distinct Cases \cite{kovenock2012coalitional}. Four of these Cases address the setting where $\frac{\oxone}{\ophione} \leq \frac{\oxtwo}{\ophitwo}$, meaning that the ratio of Player 1's budget to contest valuation (after a transfer has been performed) is less than that of Player 2; the remaining three Cases address the setting where $\frac{\oxone}{\ophione} > \frac{\oxtwo}{\ophitwo}$ and can be obtained by a simple swapping of indices. Hence, without loss of generality, we describe below the Adversary's behavior for the four Cases that comprehensively describe all games
\begin{equation}
    \OG \in \G_{1 \leq 2} \triangleq \lcurl{G \in \G \suchthat \frac{\oxone}{\ophione} \leq \frac{\oxtwo}{\ophitwo}}
\end{equation}
as illustrated in Figure \ref{fig:cases}.

\begin{figure}
    \centering
    \includegraphics[width=\linewidth]{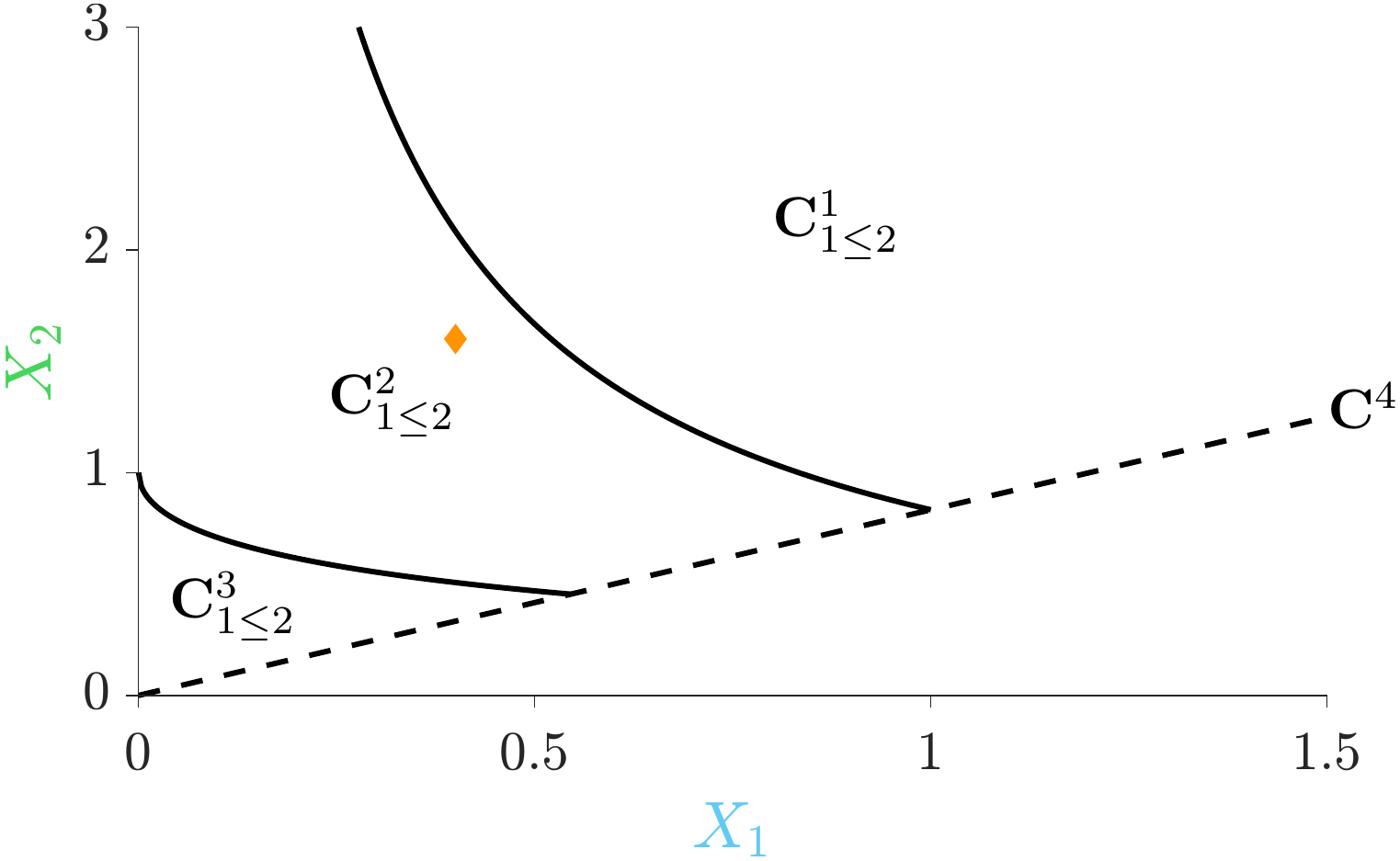}
    \caption{The Cases that delineate the form of the Adversary's best response, plotted in the $X_1$-$X_2$ space for fixed $\phi_1 = 12$, $\phi_2 = 10$. The game $G^1 = (12, 10, 0.4, 1.6)$ (Figure \ref{fig:simple_example}) is indicated by $\orangelozenge$.}
    \label{fig:cases}
\end{figure}

\vspace{.2cm}
\noindent\emph{Case 1:} If $\OG$ belongs to
    \begin{equation}\label{eq:case_1}
        \C_{1 \leq 2}^1 \triangleq \lcurl{\OG \in \G_{1 \leq 2} \suchthat \frac{\oxone}{\ophione} \neq \frac{\oxtwo}{\ophitwo} \text{ and } \frac{\ophitwo}{\ophione} \leq \oxone \oxtwo},
    \end{equation}
    then according to \eqref{eq:adversary_payoff}, every unit of budget that the Adversary allocates towards Lotto game 1 has a greater marginal payoff than if it were allocated to Lotto game 2. Hence, the Adversary allocates their entire budget towards game 1 ($\xai{1} = 1$).

\vspace{.2cm}
\noindent\emph{Case 2:} If $\OG$ belongs to
    \begin{equation}\label{eq:case_2}
        \C_{1 \leq 2}^2 \triangleq \lcurl{\OG \in \G_{1 \leq 2} \suchthat 0 < 1 - \mysqrt{\frac{\oxone \oxtwo \ophione}{\ophitwo}} \leq \oxtwo},
    \end{equation}
    then the Adversary's budget is sufficiently large such that it is optimal to allocate $\xai{1} \geq \oxone$, thereby hitting the range of diminishing returns (see \eqref{eq:adversary_payoff}) in Lotto game 1. Thus, they allocate budget towards Lotto game 2 such that the marginal payoffs in each game are equal, yielding $\xai{1} = \mysqrt{\frac{\oxone \oxtwo \ophione}{\ophitwo}}$.

\vspace{.2cm}
\noindent\emph{Case 3:} If $\OG$ belongs to
    \begin{equation}\label{eq:case_3}
        \C_{1 \leq 2}^3 \triangleq \lcurl{\OG \in \G_{1 \leq 2} \suchthat 1 - \mysqrt{\frac{\oxone \oxtwo \ophione}{\ophitwo}} > \oxtwo},
    \end{equation}
    then the Adversary's budget is larger than both Players' combined (i.e., $X_1 + X_2 \leq 1$). Thus, the Adversary equates the marginal payoffs in the range of diminishing returns in both games, which yields $\xai{1} = \frac{\mysqrt{\oxone \ophione}}{\mysqrt{\oxone \ophione} + \mysqrt{\oxtwo \ophitwo}}$.

\vspace{.2cm}
\noindent\emph{Case 4:} If $\OG$ belongs to
    \begin{equation}\label{eq:case_4}
        \C^4 \triangleq \lcurl{\OG \in \G_{1 \leq 2} \suchthat \frac{\oxone}{\ophione} = \frac{\oxtwo}{\ophitwo} \text{ and } X_1 + X_2 \geq 1},
    \end{equation}
    then the marginal payoff in the two games is already equal. Thus, the Adversary is indifferent between the two games, and any allocation satisfying $\xai{i} \leq \overline{X}_i$, $i \in \{1, 2\}$ is optimal.

\subsection{Proof of Theorem \ref{thm:mut_ben_diff}}\label{app:contest_transfer_proof}

In this section, we provide the proof of Theorem \ref{thm:mut_ben_diff}, which first requires characterizing $\G^\nu$, the subset of $\G$ in which mutually beneficial contest transfers exist. After characterizing the region, the statements in the Theorem follow immediately. Before delving into the characterization, which is lengthy and tedious, we provide some high-level intuition to the reader. As a reminder, a contest transfer is mutually beneficial if it causes both Players' payoffs to increase. Below, we show that mutually beneficial transfers can be either
\begin{itemize}[leftmargin=*]
    \item \emph{strategically consistent}, meaning that the structure of the Adversary's best response does not change in response to the transfer (i.e., $G$ and $\OG$ belong to the same Case);
    \item \emph{strategically inconsistent}, meaning that the structure of the Adversary's best response does change in response to the transfer (i.e., $G$ and $\OG$ belong to different Cases).
\end{itemize}

To present these results, we introduce some additional notation that simplifies the forthcoming discussion. Let $\C_{1 > 2}^1$, $\C_{1 > 2}^2$, and $\C_{1 > 2}^3$ denote the counterparts of $\C_{1 \leq 2}^1$, $\C_{1 \leq 2}^2$, and $\C_{1 \leq 2}^3$ (i.e., the same as \eqref{eq:case_1}, \eqref{eq:case_2}, and \eqref{eq:case_3}, respectively, with the indices swapped). Further, define the subsets
\begin{align*}
    \RR^1 &= \{ G \in \G \ \big\vert \ X_1 \geq 1, \ X_2 \geq 1 \} \\
    \RR^2 &= \{ G \in \G \ \big\vert \ X_1 \geq 1, \ X_2 < 1 \} \\
    \RR^3 &= \{ G \in \G \ \big\vert \ X_1 < 1, \ X_2 \geq 1 \} \\
    \RR^4 &= \{ G \in \G \ \big\vert \ X_1 < 1, \ X_2 < 1, X_1 + X_2 \geq 1 \} \\
    \RR^5 &= \{ G \in \G \ \big\vert \ X_1 + X_2 < 1 \}
\end{align*}
which partition the parameter space, as shown in Figure \ref{fig:regions}; these partitions determine whether $\C_{1 \leq 2}^2$, $\C_{1 \leq 2}^3$, $\C_{1 > 2}^2$, and $\C_{1 > 2}^3$ are nonempty. To expedite the forthcoming discussion, we also introduce Table \ref{table:payoffs}, which summarizes the Players' payoffs in each possible Case; here, we index the Players using $i$ and $-i$, where $-i = 2$ if $i = 1$ and $-i = 1$ if $i = 2$, which allows us to discuss and exploit the symmetry between the sets $G_{1 \leq 2}$ and $\G_{1 > 2} \triangleq \G \setminus \G_{1 \leq 2}$.

\begin{figure}
    \centering
    \includegraphics[width=\linewidth]{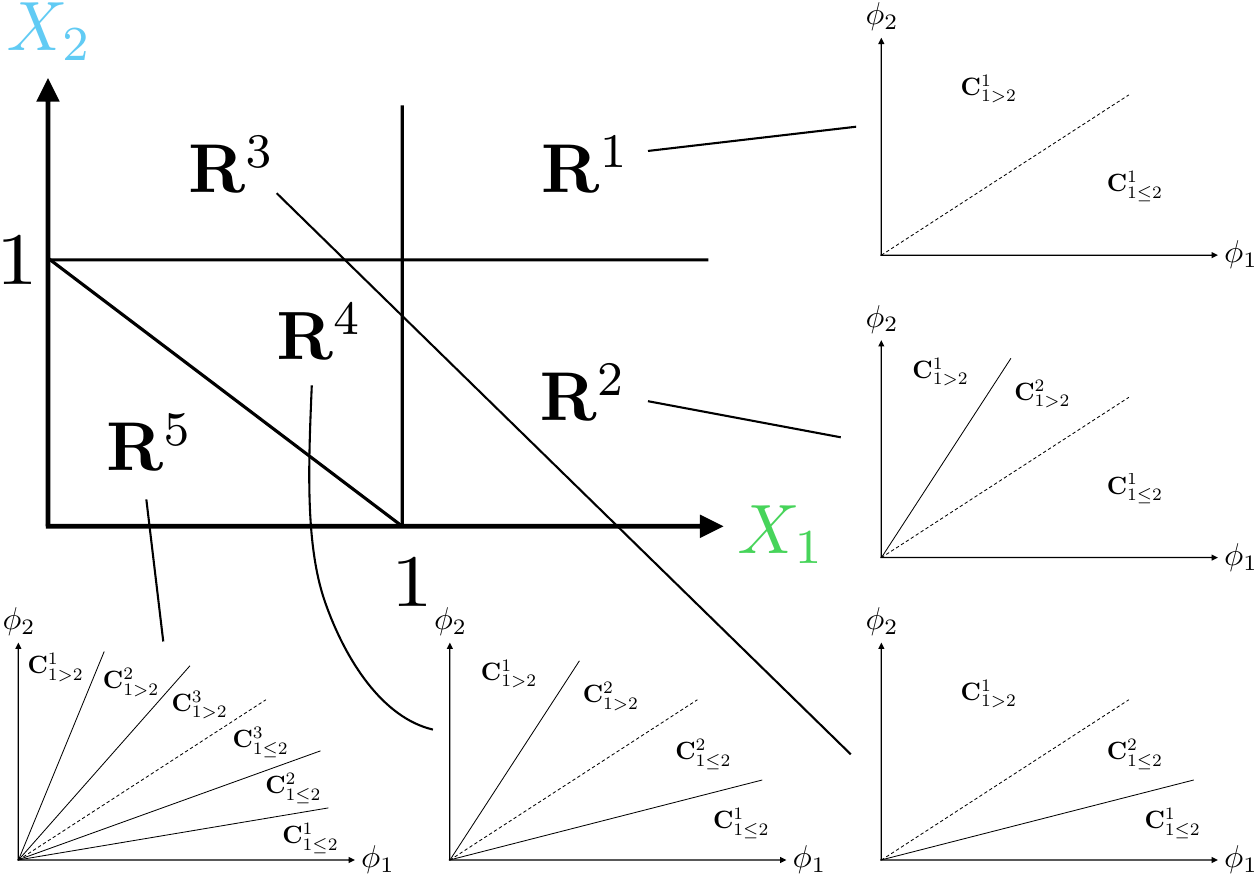}
    \caption{Cartoon depiction subsets of the $X_1$-$X_2$ space (top left, large) that determine whether certain Cases are empty, as depicted each of the five smaller plots in the $\phi_1$-$\phi_2$ space. For example, in subset $\RR^3$, the subsets $\C_{1 \leq 2}^2$, $\C_{1 \leq 2}^3$, and $\C_{1 > 2}^3$ are empty.}
    \label{fig:regions}
\end{figure}

\newcounter{peq}
\newcounter{peqgroup}
\newcounter{peqsub}

\renewcommand{\thepeq}{P\thepeqgroup\alph{peqsub}}

\newcommand{\newpeqgroup}{%
  \stepcounter{peqgroup}%
  \setcounter{peqsub}{0}%
}

\newcommand{\peqlabel}[1]{%
  \stepcounter{peqsub}%
  \refstepcounter{peq}
  \label{#1}%
  (\thepeq)%
}

\begin{table*}[t]
    \centering
    \begin{tabular}{c | c c | c c}
        Case & \multicolumn{2}{c|}{Player $i$ payoff, $U_i(0, \nu \suchthat G)$} & \multicolumn{2}{c}{Player $-i$ payoff, $U_{-i}(0, \nu \suchthat G)$} \\
        \hline \newpeqgroup
        %
        %
        $\OG \in \C_{i < -i}^1$ &
        $\begin{cases} \ophi_i \left( \frac{X_i}{2} \right) & X_i < 1 \\ \ophi_i \left( 1 - \frac{1}{2 X_i} \right) & X_i \geq 1 \end{cases}$ & \peqlabel{eq:p1a} &
        $\ophi_{-i}$ & \peqlabel{eq:p1b} \\ \newpeqgroup
        %
        %
        $\OG \in \C_{i < -i}^2$ &
        $\frac{1}{2} \mysqrt{\frac{X_i \ophi_i \ophi_{-i}}{X_{-i}}}$ & \peqlabel{eq:p2a} &
        $\ophi_{-i} \left( 1 - \frac{1}{2 X_{-i}} \right) + \frac{1}{2} \mysqrt{\frac{X_i \ophi_i \ophi_{-i}}{X_{-i}}}$ & \peqlabel{eq:p2b} \\ \newpeqgroup
        %
        %
        $\OG \in \C_{i < -i}^3$ & $\frac{1}{2} X_i \ophi_i + \mysqrt{X_i X_{-i} \ophi_i \ophi_{-i}}$ & \peqlabel{eq:p3a} & $\frac{1}{2} X_{-i} \ophi_{-i} + \mysqrt{X_i X_{-i} \ophi_i \ophi_{-i}}$  & \peqlabel{eq:p3b} \\ \newpeqgroup
        %
        %
        $\OG \in \C^4$ & $\ophi_i \lpar{1 - \frac{\xai{i}}{X_i}}$, $\xai{i} \leq X_i$ & \peqlabel{eq:p4a} & $\ophi_{-i} \lpar{1 - \frac{\xai{-i}}{X_{-i}}}$, $\xai{-i} \leq X_{-i}$  & \peqlabel{eq:p4b}
    \end{tabular}
    \medskip
    \caption{Players' payoffs for different cases. The dependence on $\nu$ is suppressed for brevity.}
    \label{table:payoffs}
\end{table*}

\subsubsection{Strategically consistent transfers}

We begin by characterizing when strategically consistent mutually beneficial contest transfers exist. First, consider any game $G \in \C_{1 \leq 2}^1$, and suppose that the Players performed a strategically consistent contest transfer, i.e., a transfer $\nu$ such that $\OG \in \C_{1 \leq 2}^1$. In this case, Player 1's payoff \eqref{eq:p1a} decreases for any $\nu > 0$, so there is no strategically consistent mutually beneficial contest transfer.

Next, consider any game $G \in \C_{1 \leq 2}^2$, and suppose that the Players performed a strategically consistent contest transfer $\nu$. In this case, Player 1 and 2's payoffs are given by \eqref{eq:p2a} and \eqref{eq:p2b}, respectively, and their derivatives are given by
\begin{align*}
    \dnu{U_1} &= \frac{1}{4} \lpar{\frac{X_1}{X_2 (\pq - \nu)(\pw + \nu)}}^{-\frac{1}{2}} (\pq - \pw - 2 \nu), \\
    \dnu{U_2} &= \left( 1 - \frac{1}{2 X_2} \right) + \dnu{U_1}.
\end{align*}
It is easily verified that $\dnu{U_1}|_{\nu = 0} > 0$ if and only if
\begin{gather}
    \pq > \pw, \label{eq:sc_c2_1}
    \intertext{and similarly that $\dnu{U_2}|_{\nu = 0} > 0$ if and only if}
    \frac{2 - 4 X_2}{\pq - \pw} < \mysqrt{\frac{X_1 X_2}{\pq \pw}}. \label{eq:sc_c2_2}
\end{gather}

Next, consider any game $G \in \C_{1 \leq 2}^3$, and suppose that the Players performed a strategically consistent contest transfer $\nu$. In this case, Player 1 and 2's payoffs are given by \eqref{eq:p3a} and \eqref{eq:p3b}, respectively, and their derivatives are given by
\begin{align*}
    \dnu{U_1} &= -\frac{1}{2} X_1 + \mysqrt{\frac{X_1 X_2}{(\pq - \nu) (\pw + \nu)}} (\pq - \pw - 2 \nu), \\
    \dnu{U_2} &= \frac{1}{2} X_2 + \mysqrt{\frac{X_1 X_2}{(\pq - \nu) (\pw + \nu)}} (\pq - \pw - 2 \nu).
\end{align*}
Since $\dnu{U_2} > \dnu{U_1}$, a sufficient condition for the existence of a positive mutually beneficial contest transfer is $\dnu{U_1} \vert_{\nu = 0} > 0$, which can be rewritten as
\begin{equation}\label{eq:sc_c3}
    X_2 > \frac{4 \pq \pw}{(\pq - \pw)^2} X_1
\end{equation}
with $\pq > \pw$ necessarily.

Finally, for any game $G \in \C^4$, the sum of the Players' payoffs is already maximized, so no transfer can strictly increase both Players' payoffs. Thus, we define
\begin{equation}
\begin{split}
    \G_{1 \leq 2}^{\rm SC} = &\lcurl{G \in \C_{1 \leq 2}^2 \ \big\vert \ \text{$G$ satisfies \eqref{eq:sc_c2_1} and \eqref{eq:sc_c2_2}}} \\
    \cup \, &\lcurl{G \in \C_{1 \leq 2}^3 \ \big\vert \ \text{$G$ satisfies \eqref{eq:sc_c3}}}
\end{split}
\end{equation}
as the set of game instances in $\G_{1 \leq 2}$ in which Players can perform strategically consistent mutually beneficial contest transfers.

\subsubsection{Strategically inconsistent transfers}

In this section, we consider strategically inconsistent transfers, i.e., transfers that cause the form of the Adversary's best response to change. Depending on the Players' initial budgets, a transfer can induce different forms of best responses from the Adversary. Figure \ref{fig:regions} describes the different Cases that can be reached via a contest transfer depending on the Players' initial budgets. We proceed by analyzing each possible transfer in each of the five Regions. We present a thorough analysis for a few representative kinds of transfers; the remaining cases are similar, so we omit several steps for brevity.

\noindent \textbf{Region 1} $(\mathbf{X_1 \geq 1, \, X_2 \geq 1})$: There is one possible form of strategically inconsistent transfer in this Region. Suppose that the Players performed a transfer $\nu$ such that $\OG \in \C_{1 > 2}^1$. In order for this to occur, the transfer must be sufficiently large so that the Adversary allocates the entirety of their budget towards the game against Player 2, meaning that $\nu$ must satisfy
\begin{equation}\label{eq:cross_middle_line}
    \frac{X_1}{\pq - \nu} > \frac{X_2}{\pw + \nu} \iff \nu > \frac{X_2 \pq - X_1 \pw}{X_1 + X_2} \triangleq \alpha_1.
\end{equation}
For the transfer to be mutually beneficial, \eqref{eq:mut_ben_transfer} must be satisfied. The pre-transfer payoffs $U_1(\tau, 0 \suchthat G)$ and $U_2(\tau, 0 \suchthat G)$ are given by \eqref{eq:p1a} and \eqref{eq:p1b} with $i = 1$ and $\nu = 0$, and the post-transfer payoffs are given by \eqref{eq:p1b} and \eqref{eq:p1a} with $i = 2$. Substituting these expressions into \eqref{eq:mut_ben_transfer}, we have
\begin{align}
    \pq \lpar{\frac{X_1}{2}} &< \pq - \nu \label{eq:r1_c1a_c1b_1} \\
    \text{and} \; \pw &< (\pw + \nu) \left( \frac{X_2}{2} \right). \label{eq:r1_c1a_c1b_2}
\end{align}
With some minor algebraic manipulation, it can be shown that there exists a transfer satisfying \eqref{eq:cross_middle_line}, \eqref{eq:r1_c1a_c1b_1}, and \eqref{eq:r1_c1a_c1b_2} if and only if
\begin{equation}\label{eq:r1_c1a_c1b_cond}
    \frac{2 X_1 X_2 - X_1 - X_2}{2 X_1^2} < \frac{\pw}{\pq} < \frac{2 X_2 - 1}{2 X_1}.
\end{equation}
Thus, we define
\begin{equation}
    \G_{1 \leq 2}^{\rm SI, 1} = \lcurl{G \in \C_{1 \leq 2}^1 \, \cap \, \RR^1 \ \big\vert \ \text{$G$ satisfies \eqref{eq:r1_c1a_c1b_cond}}}
\end{equation}
as the set of game instances in Region 1 in which Players can perform positive, strategically inconsistent, and mutually beneficial contest transfers.

\newcounter{regsec}
\setcounter{regsec}{1}
\newcounter{formsec}
\newcounter{formref} 

\renewcommand{\theformref}{\theregsec.\theformsec}

\newcommand{\mylabel}{\textbf{\theformref}\ }

\newcommand{\newregsec}{%
  \stepcounter{regsec}%
  \setcounter{formsec}{0}%
}

\newcommand{\myformsec}[1]{%
  \par\noindent
  \stepcounter{formsec}%
  \refstepcounter{formref}%
  \label{#1}%
  \mylabel
}

\newregsec
\noindent \textbf{Region 2} $(\mathbf{X_1 \geq 1, \, X_2 < 1})$: There are two possible forms of strategically inconsistent transfers in this Region. 

\myformsec{formsec:r2_1a_2b} $\lpar{\C_{1 \leq 2}^1 \to \C_{1 > 2}^2}$: In order for $\OG \in \C_{1 > 2}^2$, $\nu$ must satisfy \eqref{eq:cross_middle_line} and
\begin{equation}\label{eq:stay_in_2b}
    \frac{\pq - \nu}{\pw + \nu} > X_1 X_2 \iff \nu < \frac{\pq - X_1 X_2 \pw}{X_1 X_2 + 1} \triangleq \alpha_2.
\end{equation}
Player 1 and 2's pre-transfer payoffs are given by \eqref{eq:p1a} and \eqref{eq:p1b} with $i = 1$ and $\nu = 0$, and their post-transfer payoffs are given by \eqref{eq:p2b} and \eqref{eq:p2a} with $i = 2$, respectively. With these payoffs, \eqref{eq:mut_ben_transfer} can be written as two quadratic inequalities in $\nu$ of the form
\begin{equation}\label{eq:quadratic_inequality}
    a_i \nu^2 + b_i \nu + c_i < 0
\end{equation}
with
\begin{gather}
    a_1 = 1 + \frac{(2 X_1 - 1)^2}{X_1 X_2}, \ b_1 = \pw - \pq, \ c_1 = -\pq \pw \label{eq:r2_c1a_c2b_1} \\
    \text{and } a_2 = 1, \ b_2 = \pw - \pq, \ c_2 = 4 \frac{X_1}{X_2} \pw^2 - 4 \pq \pw. \label{eq:r2_c1a_c2b_2}
\end{gather}
Let $d_1$ and $d_2$ denote the discriminants of the quadratics defined by \eqref{eq:r2_c1a_c2b_1} and \eqref{eq:r2_c1a_c2b_2}, respectively, where $$d_i  \triangleq b_i^2 - 4 a_i c_i.$$ Similarly, let $z_i^-$ and $z_i^+$ denote the solutions of the quadratics defined by $a_i$, $b_i$, and $c_i$, i.e., $$z_i^{\pm} \triangleq \frac{-b \pm \sqrt{d_i}}{2 a_i}.$$ An inequality of the form \eqref{eq:quadratic_inequality} is satisfied if and only if $d_i > 0$ and $z_i^- < \nu < z_i^+$. Thus, there exists a transfer satisfying \eqref{eq:mut_ben_transfer}, \eqref{eq:cross_middle_line}, and \eqref{eq:stay_in_2b} if and only if
\begin{equation}\label{eq:r2_c1a_c2b_cond}
  \begin{gathered}
    d_1 > 0, \, d_2 > 0,  \\
    \text{and} \, \max \lcurl{z_1^-, z_2^-, \alpha_1} < \min \lcurl{z_1^+, z_2^+, \alpha_2}.
  \end{gathered}
\end{equation}

\myformsec{formsec:r2_1a_1b} $\lpar{\C_{1 \leq 2}^1 \to \C_{1 > 2}^1}$: In order for $\OG \in \C_{1 > 2}^1$, $\nu$ must satisfy
\begin{equation}\label{eq:go_to_1b}
    \frac{\pq - \nu}{\pw + \nu} \leq X_1 X_2 \iff \nu \geq \alpha_2.
\end{equation}
Player 1 and 2's pre-transfer payoffs are given by \eqref{eq:p1a} and \eqref{eq:p1b} with $i = 1$ and $\nu = 0$, and their post-transfer payoffs are given by \eqref{eq:p1b} and \eqref{eq:p1a} with $i = 2$, respectively. With these payoffs, it can be shown that there exists a transfer satisfying \eqref{eq:mut_ben_transfer} and \eqref{eq:go_to_1b} if and only if
\begin{equation}\label{eq:r2_c1a_c1b_cond}
    \frac{-X_1 X_2 + 2X_1 - 1}{2 X_1^2 X_2} < \frac{\pw}{\pq} < \frac{X_2}{2 X_1 (2 - X_2)}.
\end{equation}

Thus, we define
\begin{equation}
    \G_{1 \leq 2}^{\rm SI, 2} = \lcurl{G \in \C_{1 \leq 2}^1 \, \cap \, \RR^2 \ \big\vert \ \text{$G$ satisfies \eqref{eq:r2_c1a_c1b_cond} or \eqref{eq:r2_c1a_c2b_cond}}}
\end{equation}
as the set of game instances in Region 2 in which Players can perform positive, strategically inconsistent, and mutually beneficial contest transfers.

For the remaining Regions, the procedure is similar: one must establish conditions on the transfer so that the induced game $\OG$ belongs to the appropriate case, and one must also establish conditions so that \eqref{eq:mut_ben_transfer} is satisfied. Henceforth, we omit much of the discussion and simply state the conditions, which follow from the same steps as shown above.

\newregsec
\noindent \textbf{Region 3} $(\mathbf{X_1 < 1, \, X_2 \geq 1})$: There are three possible forms of strategically inconsistent transfers in this Region. 

\myformsec{formsec:r3_1a_2a} $\lpar{\C_{1 \leq 2}^1 \to \C_{1 \leq 2}^2}$: In order for $\OG \in \C_{1 \leq 2}^2$, $\nu$ must satisfy
\begin{gather}
    \frac{\pw + \nu}{\pq - \nu} > X_1 X_2 \iff \nu > \frac{X_1 X_2 \pq - \pw}{X_1 X_2 + 1} \triangleq \alpha_3 \label{eq:1a_to_2a} \\
    \text{and} \quad
    \frac{X_1}{\pq - \nu} > \frac{X_2}{\pw + \nu} \iff \nu < \alpha_1. \label{eq:no_cross_middle_line}
\end{gather}
The pre-transfer payoffs $U_1$ and $U_2$ are given by \eqref{eq:p1a} and \eqref{eq:p1b} with $i = 1$ and $\nu = 0$, and the post-transfer payoffs are given by \eqref{eq:p2a} and \eqref{eq:p2b} with $i = 1$. The first condition in \eqref{eq:mut_ben_transfer}, $U_1(\tau, 0 \suchthat G) < U_1(\tau, \nu \suchthat G)$, can be written in the form \eqref{eq:quadratic_inequality} with
\begin{equation}\label{eq:r3_1a_1b_u1_cond}
    a_3 = 1, \  b_3 = \pw - \pq, \ c_3 = X_1 X_2 \pq^2 - \pq \pw,
\end{equation}
and it is easily verified that for any positive transfer $\nu$, $U_1(\tau, 0 \suchthat G) < U_1(\tau, \nu \suchthat G) \implies U_2(\tau, 0 \suchthat G) < U_2(\tau, \nu \suchthat G)$. Thus, there exists a positive transfer satisfying \eqref{eq:mut_ben_transfer}, \eqref{eq:1a_to_2a} and \eqref{eq:no_cross_middle_line} if and only if
\begin{equation}\label{eq:r3_c1a_c2a_cond}
    d_3 > 0 \; \text{and} \, \max \lcurl{z_3^-, \alpha_3} < \min \lcurl{z_3^+, \alpha_1}.
\end{equation}

\myformsec{formsec:r3_1a_1b} $\lpar{\C_{1 \leq 2}^1 \to \C_{1 > 2}^1}$: In order for $\OG \in \C_{1 > 2}^1$, $\nu$ must satisfy \eqref{eq:cross_middle_line}. The pre-transfer payoffs $U_1$ and $U_2$ are given by \eqref{eq:p1a} and \eqref{eq:p1b} with $i = 1$ and $\nu = 0$, and the post-transfer payoffs are given by \eqref{eq:p1b} and \eqref{eq:p1a} with $i = 2$. There exists a transfer satisfying \eqref{eq:mut_ben_transfer} and \eqref{eq:cross_middle_line} if and only if
\begin{equation}\label{eq:r3_c1a_c1b_cond}
    \frac{1}{2} (X_1 + X_2 - 2) < \frac{\pw}{\pq} < \frac{1}{2} (-X_1 + 2)(2 X_2 - 1).
\end{equation}

\myformsec{formsec:r3_2a_1b} $\lpar{\C_{1 \leq 2}^2 \to \C_{1 > 2}^1}$: In order for $\OG \in \C_{1 > 2}^1$, $\nu$ must satisfy \eqref{eq:cross_middle_line}. The pre-transfer payoffs $U_1$ and $U_2$ are given by \eqref{eq:p2a} and \eqref{eq:p2b} with $i = 1$ and $\nu = 0$, and the post-transfer payoffs are given by \eqref{eq:p1b} and \eqref{eq:p1a} with $i = 2$. There exists a transfer satisfying \eqref{eq:mut_ben_transfer} and \eqref{eq:cross_middle_line} if and only if
\begin{equation}\label{eq:r3_c2a_c1b_cond}
    \max \lcurl{\alpha_1, \frac{\mysqrt{X_1 X_2 \pq \pw}}{2 X_2 - 1} } < \pq - \frac{1}{2} \mysqrt{\frac{X_1 \pw \pw}{X_2}}.
\end{equation}
Thus, we define
\begin{equation}
\begin{split}
    \G_{1 \leq 2}^{\rm SI, 3} = &\lcurl{G \in \C_{1 \leq 2}^1 \, \cap \, \RR^3 \ \big\vert \ \text{$G$ satisfies \eqref{eq:r3_c1a_c2a_cond} or \eqref{eq:r3_c1a_c1b_cond}}} \\
    \cup \, &\lcurl{G \in \C_{1 \leq 2}^2 \, \cap \, \RR^3 \ \big\vert \ \text{$G$ satisfies \eqref{eq:r3_c2a_c1b_cond}}}
\end{split}
\end{equation}
as the set of game instances in Region 3 in which Players can perform positive, strategically inconsistent, and mutually beneficial contest transfers.

\newregsec
\noindent \textbf{Region 4} $(\mathbf{X_1 < 1, \, X_2 \geq 1})$: There are five possible forms of strategically inconsistent transfers in this Region.

\myformsec{formsec:r4_1a_2a} $\lpar{\C_{1 \leq 2}^1 \to \C_{1 \leq 2}^2}$: In order for $\OG \in \C_{1 \leq 2}^2$, $\nu$ must satisfy \eqref{eq:1a_to_2a} and \eqref{eq:no_cross_middle_line}. The pre-transfer payoffs $U_1$ and $U_2$ are given by \eqref{eq:p1a} and \eqref{eq:p1b} with $i = 1$ and $\nu = 0$, and the post-transfer payoffs are given by \eqref{eq:p2a} and \eqref{eq:p2b} with $i = 1$. The condition $U_1(\tau, 0 \suchthat G) < U_1(\tau, \nu \suchthat G)$ can again be written in the form \eqref{eq:quadratic_inequality} with parameters given by \eqref{eq:r3_1a_1b_u1_cond}, and it is readily verified that $U_1(\tau, 0 \suchthat G) < U_1(\tau, \nu \suchthat G) \implies U_2(\tau, 0 \suchthat G) < U_2(\tau, \nu \suchthat G)$ when $X_2 \geq \frac{1}{2}$. When $X_2 < \frac{1}{2}$, the condition $U_2(\tau, 0 \suchthat G) < U_2(\tau, \nu \suchthat G)$ can be written in the form \eqref{eq:quadratic_inequality} with
\begin{gather*}
    a_4 = (1 - 2 X_2)^2 + X_1 X_2, \; c_4 = \pw^2 - X_1 X_2 \pq \pw, \\
    b_4 = 2 (1 - 2 X_2) \pw + X_1 X_2 (\pw - \pq).
\end{gather*}
Thus, there exists a positive transfer satisfying \eqref{eq:mut_ben_transfer}, \eqref{eq:1a_to_2a} and \eqref{eq:no_cross_middle_line} if and only if
\begin{equation}\label{eq:r4_c1a_c2a_cond}
  \begin{gathered}
  d_3 > 0, X_2 \geq \frac{1}{2}, \, \text{and} \, \max \lcurl{z_3^-, \alpha_3} < \min \lcurl{z_3^+, \alpha_1} \\
    \text{or} \; d_3 > 0, d_4 > 0, X_2 < \frac{1}{2}, \\
    \text{and} \, \max \lcurl{z_3^-, z_4^-, \alpha_3} < \min \lcurl{z_3^+, z_4^+, \alpha_1}.
  \end{gathered}
\end{equation}

\myformsec{formsec:r4_1a_2b} $\lpar{\C_{1 \leq 2}^1 \to \C_{1 > 2}^2}$: In order for $\OG \in \C_{1 > 2}^2$, $\nu$ must satisfy \eqref{eq:cross_middle_line} and \eqref{eq:stay_in_2b}. The pre-transfer payoffs $U_1$ and $U_2$ are given by \eqref{eq:p1a} and \eqref{eq:p1b} with $i = 1$ and $\nu = 0$, and the post-transfer payoffs are given by \eqref{eq:p2b} and \eqref{eq:p2a} with $i = 2$. The conditions \eqref{eq:mut_ben_transfer} can be written in the form \eqref{eq:quadratic_inequality} with
\begin{gather*}
    a_5 = (1 - 2 X_1)^2 + X_1 X_2, \; c_5 = (X_1 - 1)^4 \pq^2 - X_1 X_2 \pq \pw, \\
    b_5 = X_1 X_2 (\pw - \pq) - 2 (X_1 - 1)^2 (1 - 2 X_1) \pq \\
    \text{and} \;
    a_6 = 1, \; b_6 = \pw - \pq, \; c_6 = 4 \frac{X_1}{X_2} \pw^2 - \pq \pw
\end{gather*}
for $U_1$ and $U_2$, respectively. Thus, there exists a positive transfer satisfying \eqref{eq:mut_ben_transfer}, \eqref{eq:cross_middle_line}, and \eqref{eq:stay_in_2b} if and only if
\begin{equation}\label{eq:r4_c1a_c2b_cond}
  \begin{gathered}
    d_5 > 0, d_6 > 0, \\
    \text{and} \, \max \lcurl{z_5^-, z_6^-, \alpha_1} < \min \lcurl{z_5^+, z_6^+, \alpha_2}.
  \end{gathered}
\end{equation}

\myformsec{formsec:r4_1a_1b} $\lpar{\C_{1 \leq 2}^1 \to \C_{1 > 2}^1}$: In order for $\OG \in \C_{1 > 2}^1$, $\nu$ must satisfy \eqref{eq:go_to_1b}. The pre-transfer payoffs $U_1$ and $U_2$ are given by \eqref{eq:p1b} and \eqref{eq:p1a} with $i = 2$ and $\nu = 0$, and the post-transfer payoffs are given by \eqref{eq:p1a} and \eqref{eq:p1b} with $i = 1$. There exists a transfer satisfying \eqref{eq:mut_ben_transfer} and \eqref{eq:go_to_1b} if and only if
\begin{equation}\label{eq:r4_c1a_c1b_cond}
    \frac{X_1 X_2 - 2 X_2 + 1}{2 X_2} < \frac{\pw}{\pq} < \frac{(-X_1 + 2)X_2}{2(-X_2 + 2)}.
\end{equation}

\myformsec{formsec:r4_2a_2b} $\lpar{\C_{1 \leq 2}^2 \to \C_{1 > 2}^2}$: In order for $\OG \in \C_{1 > 2}^2$, $\nu$ must satisfy \eqref{eq:cross_middle_line} and \eqref{eq:stay_in_2b}. The pre-transfer payoffs $U_1$ and $U_2$ are given by \eqref{eq:p2a} and \eqref{eq:p2b} with $i = 1$ and $\nu = 0$, and the post-transfer payoffs are given by \eqref{eq:p2b} and \eqref{eq:p2a} with $i = 2$. The conditions \eqref{eq:mut_ben_transfer} can be written in the form \eqref{eq:quadratic_inequality} with
\begin{gather*}
    a_7 = (2 X_1 - 1)^2 + X_1 X_2, \\
    b_7 = 2 (2 X_1 - 1) \lpar{X_1 \mysqrt{\frac{X_1 \pw \pw}{X_2}} - (2 X_1 - 1) \pq} \\ + X_1 X_2(\pw - \pq), \\
    c_7 = \lpar{X_1 \mysqrt{\frac{X_1 \pw \pw}{X_2}} - (2 X_1 - 1) \pq}^2 - X_1 X_2 \pq \pw \\
    \text{and} \;
    a_8 = 1, \; b_8 = \pw - \pq, \\
    c_8 = \frac{X_1}{X_2} \lpar{\lpar{2 - \frac{1}{X_2}} \pw + \mysqrt{\frac{X_1 \pq \pw}{X_2}} }^2 - \pq \pw
\end{gather*}
for $U_1$ and $U_2$, respectively. Thus, there exists a positive transfer satisfying \eqref{eq:mut_ben_transfer}, \eqref{eq:cross_middle_line}, and \eqref{eq:stay_in_2b} if and only if
\begin{equation}\label{eq:r4_c2a_c2b_cond}
  \begin{gathered}
    d_7 > 0, d_8 > 0, \\
    \text{and} \, \max \lcurl{z_7^-, z_8^-, \alpha_1} < \min \lcurl{z_7^+, z_8^+, \alpha_2}.
  \end{gathered}
\end{equation}

\myformsec{formsec:r4_2a_1b} $\lpar{\C_{1 \leq 2}^2 \to \C_{1 > 2}^1}$: In order for $\OG \in \C_{1 > 2}^1$, $\nu$ must satisfy \eqref{eq:go_to_1b}. The pre-transfer payoffs $U_1$ and $U_2$ are given by \eqref{eq:p2a} and \eqref{eq:p2b} with $i = 1$ and $\nu = 0$, and the post-transfer payoffs are given by \eqref{eq:p1b} and \eqref{eq:p1a} with $i = 2$. There exists a transfer satisfying \eqref{eq:mut_ben_transfer} and \eqref{eq:go_to_1b} if and only if
\begin{equation}\label{eq:r4_c2a_c1b_cond}
\begin{gathered}
    \max \lcurl{\alpha_2, -\frac{(X_2 - 1)^2}{X_2^2}\pw + \mysqrt{\frac{X_1 \pq \pw}{X_2^3} }} \\
    < \pq - \frac{1}{2} \mysqrt{\frac{X_1 \pw \pw}{X_2}}.
\end{gathered}
\end{equation}
Thus, we define
\begin{equation}
\begin{split}
    \G_{1 \leq 2}^{\rm SI, 4} = &\lcurl{G \in \C_{1 \leq 2}^1 \, \cap \, \RR^4 \ \big\vert \ \text{$G$ satisfies \eqref{eq:r4_c1a_c2a_cond}, \eqref{eq:r4_c1a_c2b_cond}, or \eqref{eq:r4_c1a_c1b_cond}}} \\
    \cup \, &\lcurl{G \in \C_{1 \leq 2}^2 \, \cap \, \RR^4 \ \big\vert \ \text{$G$ satisfies \eqref{eq:r4_c2a_c2b_cond} or \eqref{eq:r4_c2a_c1b_cond}}}
\end{split}
\end{equation}
as the set of game instances in Region 4 in which Players can perform positive, strategically inconsistent, and mutually beneficial contest transfers.

\newregsec
\noindent \textbf{Region 5} $(\mathbf{X_1 + X_2 < 1})$: There are twelve possible forms of strategically inconsistent transfers in this Region.

\myformsec{formsec:r5_1a_2a} $\lpar{\C_{1 \leq 2}^1 \to \C_{1 \leq 2}^2}$: In order for $\OG \in \C_{1 \leq 2}^2$, $\nu$ must satisfy \eqref{eq:1a_to_2a} and
\begin{gather}
    1 - \mysqrt{\frac{X_1 X_2 (\pq - \nu)}{\pw + \nu}} \leq X_2 \nonumber \\
    \iff \nu \leq \frac{X_1 X_2 \pq - (1 - X_2)^2 \pw}{(1 - X_2)^2 + X_1 X_2} \triangleq \alpha_4. \label{eq:stay_in_2a}
\end{gather}
The analysis is otherwise identical to \textbf{\ref{formsec:r4_1a_2a}}. Thus, there exists a positive transfer satisfying \eqref{eq:mut_ben_transfer}, \eqref{eq:1a_to_2a} and \eqref{eq:stay_in_2a} if and only if
\begin{equation}\label{eq:r5_c1a_c2a_cond}
  \begin{gathered}
  d_3 > 0, X_2 \geq \frac{1}{2}, \, \text{and} \, \max \lcurl{z_3^-, \alpha_3} < \min \lcurl{z_3^+, \alpha_4} \\
    \text{or} \; d_3 > 0, d_4 > 0, X_2 < \frac{1}{2}, \\
    \text{and} \, \max \lcurl{z_3^-, z_4^-, \alpha_3} < \min \lcurl{z_3^+, z_4^+, \alpha_4}.
  \end{gathered}
\end{equation}

\myformsec{formsec:r5_1a_3a} $\lpar{\C_{1 \leq 2}^1 \to \C_{1 \leq 2}^3}$: In order for $\OG \in \C_{1 \leq 2}^3$, $\nu$ must satisfy
\begin{equation}\label{eq:go_to_3}
    1 - \mysqrt{\frac{X_1 X_2 (\pq - \nu)}{\pw + \nu}} > X_2 \iff \nu > \alpha_4
\end{equation}
and \eqref{eq:no_cross_middle_line}. The pre-transfer payoffs $U_1$ and $U_2$ are given by \eqref{eq:p1a} and \eqref{eq:p1b} with $i = 1$ and $\nu = 0$, and the post-transfer payoffs are given by \eqref{eq:p3a} and \eqref{eq:p3b} with $i = 1$. The condition $U_1(\tau, 0 \suchthat G) < U_1(\tau, \nu \suchthat G)$ can be written in the form \eqref{eq:quadratic_inequality} with
\begin{gather*}
    a_9 = 1 + \frac{X_1}{X_2}, \;
    b_9 = \pw - \pq, \;
    c_9 = - \pq \pw.
\end{gather*}
If $\nu \geq \frac{2 - X_2}{X_2} \pw \triangleq \beta_1$, then $U_2(\tau, 0 \suchthat G) < U_2(\tau, \nu \suchthat G)$ is trivially satisfied; otherwise, the condition can be written in the form \eqref{eq:quadratic_inequality} with
\begin{gather*}
    a_{10} = 1 + \frac{X_2}{X_1}, \; 
    b_{10} = \frac{X_1 + 2 X_2 - 4}{X_1} \pw - \pq, \\
    c_{10} = \frac{(2 - X_2)^2 \pw^2}{X_1 X_2} - \pq \pw
\end{gather*}
Thus, there exists a positive transfer satisfying \eqref{eq:mut_ben_transfer}, \eqref{eq:go_to_3}, and \eqref{eq:no_cross_middle_line} if and only if
\begin{equation}\label{eq:r5_c1a_c3a_cond}
  \begin{gathered}
    d_9 > 0 \; \text{and} \, \max \lcurl{z_9^-, \beta_1, \alpha_4} < \min \lcurl{z_9^+, \alpha_1} \\
    \text{or} \; d_9 > 0, \, d_{10} > 0, \\
    \; \text{and} \, \max \lcurl{z_9^-, z_{10}^-, \alpha_4} < \min \lcurl{z_9^+, z_{10}^+, \beta_1, \alpha_1}.
  \end{gathered}
\end{equation}

\myformsec{formsec:r5_1a_3b} $\lpar{\C_{1 \leq 2}^1 \to \C_{1 > 2}^3}$: In order for $\OG \in \C_{1 > 2}^3$, $\nu$ must satisfy \eqref{eq:cross_middle_line} and
\begin{gather}
    1 - \mysqrt{\frac{X_1 X_2 (\pw + \nu)}{\pq - \nu}} > X_1 \nonumber \\
    \iff \nu < \frac{(1 - X_1)^2 \pq - X_1 X_2 \pw}{(1 - X_1)^2 + X_1 X_2} \triangleq \alpha_5. \label{eq:stay_in_3}
\end{gather}
The analysis is otherwise identical to \textbf{\ref{formsec:r5_1a_3a}}. Thus, there exists a positive transfer satisfying \eqref{eq:mut_ben_transfer}, \eqref{eq:cross_middle_line} and \eqref{eq:stay_in_3} if and only if
\begin{equation}\label{eq:r5_c1a_c3b_cond}
  \begin{gathered}
    d_9 > 0 \; \text{and} \, \max \lcurl{z_9^-, \beta_1, \alpha_1} < \min \lcurl{z_9^+, \alpha_5} \\
    \text{or} \; d_9 > 0, \, d_{10} > 0, \\
    \; \text{and} \, \max \lcurl{z_9^-, z_{10}^-, \alpha_1} < \min \lcurl{z_9^+, z_{10}^+, \beta_1, \alpha_5}.
  \end{gathered}
\end{equation}

\myformsec{formsec:r5_1a_2b} $\lpar{\C_{1 \leq 2}^1 \to \C_{1 > 2}^2}$: In order for $\OG \in \C_{1 > 2}^2$, $\nu$ must satisfy
\begin{equation}\label{eq:go_to_2b_2}
    1 - \mysqrt{\frac{X_1 X_2 (\pw + \nu)}{\pq - \nu}} \leq X_1 \iff \nu \geq \alpha_5
\end{equation}
and \eqref{eq:stay_in_2b}. The analysis is otherwise identical to \textbf{\ref{formsec:r4_1a_2b}}. Thus, there exists a positive transfer satisfying \eqref{eq:mut_ben_transfer}, \eqref{eq:go_to_2b_2} and \eqref{eq:stay_in_2b} if and only if
\begin{equation}\label{eq:r5_c1a_c2b_cond}
  \begin{gathered}
    d_5 > 0, d_6 > 0, \\
    \text{and} \, \max \lcurl{z_5^-, z_6^-, \alpha_5} < \min \lcurl{z_5^+, z_6^+, \alpha_2}.
  \end{gathered}
\end{equation}

\myformsec{formsec:r5_1a_1b} $\lpar{\C_{1 \leq 2}^1 \to \C_{1 > 2}^1}$: The analysis is identical to \textbf{\ref{formsec:r4_1a_1b}}.

\myformsec{formsec:r5_2a_3a} $\lpar{\C_{1 \leq 2}^2 \to \C_{1 \leq 2}^3}$: In order for $\OG \in \C_{1 \leq 2}^3$, $\nu$ must satisfy \eqref{eq:go_to_3} and \eqref{eq:no_cross_middle_line}. The pre-transfer payoffs $U_1$ and $U_2$ are given by \eqref{eq:p2a} and \eqref{eq:p2b} with $i = 1$ and $\nu = 0$, and the post-transfer payoffs are given by \eqref{eq:p3a} and \eqref{eq:p3b} with $i = 1$. The condition $U_1(\tau, 0 \suchthat G) < U_1(\tau, \nu \suchthat G)$ can be written in the form \eqref{eq:quadratic_inequality} with
\begin{gather*}
    a_{11} = 1 + \frac{X_2}{X_1}, \; b_{11} = 2 \mysqrt{\frac{\pq \pw}{X_1 X_2}} + \frac{X_2}{X_1} (\pw - \pq) - 2 \pq, \\
    c_{11} = \lpar{\mysqrt{\frac{\pq \pw}{X_1 X_2}} - \pq}^2 - \frac{X_2}{X_1} \pq \pw.
\end{gather*}
If $\nu \geq \mysqrt{\frac{X_1 \pq \pw}{X_2^3}} - \frac{(1 - X_2)^2}{X_2^2} \pw \triangleq \beta_2$, then $U_2(\tau, 0 \suchthat G) < U_2(\tau, \nu \suchthat G)$ is trivially satisfied; otherwise, the condition can be written in the form \eqref{eq:quadratic_inequality} with
\begin{gather*}
    a_{12} = 1 + \frac{X_1}{X_2}, \; b_{12} = - 2 \beta_2 + \frac{X_1}{X_2} (\pw - \pq), \\
    c_{12} = \beta_2^2 - \frac{X_1}{X_2} \pq \pw.
\end{gather*}
Thus, there exists a positive transfer satisfying \eqref{eq:mut_ben_transfer}, \eqref{eq:go_to_3}, and \eqref{eq:no_cross_middle_line} if and only if
\begin{equation}\label{eq:r5_c2a_c3a_cond}
  \begin{gathered}
    d_{11} > 0 \; \text{and} \, \max \lcurl{z_{11}^-, \beta_2, \alpha_4} < \min \lcurl{z_{11}^+, \alpha_1} \\
    \text{or} \; d_{11} > 0, \, d_{12} > 0, \\
    \; \text{and} \, \max \lcurl{z_{11}^-, z_{12}^-, \alpha_4} < \min \lcurl{z_{11}^+, z_{12}^+, \beta_2, \alpha_1}.
  \end{gathered}
\end{equation}

\myformsec{formsec:r5_2a_3b} $\lpar{\C_{1 \leq 2}^1 \to \C_{1 > 2}^3}$: In order for $\OG \in \C_{1 > 2}^3$, $\nu$ must satisfy \eqref{eq:cross_middle_line} and \eqref{eq:stay_in_3}. The analysis is otherwise identical to \textbf{\ref{formsec:r5_2a_3a}}. Thus, there exists a positive transfer satisfying \eqref{eq:mut_ben_transfer}, \eqref{eq:cross_middle_line} and \eqref{eq:stay_in_3} if and only if
\begin{equation}\label{eq:r5_c2a_c3b_cond}
  \begin{gathered}
    d_{11} > 0 \; \text{and} \, \max \lcurl{z_{11}^-, \beta_2, \alpha_1} < \min \lcurl{z_{11}^+, \alpha_5} \\
    \text{or} \; d_{11} > 0, \, d_{12} > 0, \\
    \; \text{and} \, \max \lcurl{z_{11}^-, z_{12}^-, \alpha_1} < \min \lcurl{z_{11}^+, z_{12}^+, \beta_2, \alpha_5}.
  \end{gathered}
\end{equation}

\myformsec{formsec:r5_2a_2b} $\lpar{\C_{1 \leq 2}^2 \to \C_{1 > 2}^2}$: In order for $\OG \in \C_{1 > 2}^2$, $\nu$ must satisfy \eqref{eq:go_to_2b_2} and \eqref{eq:stay_in_2b}. The analysis is otherwise identical to \textbf{\ref{formsec:r4_2a_2b}}. Thus, there exists a positive transfer satisfying \eqref{eq:mut_ben_transfer}, \eqref{eq:go_to_2b_2}, and \eqref{eq:stay_in_2b} if and only if
\begin{equation}\label{eq:r5_c2a_c2b_cond}
  \begin{gathered}
    d_7 > 0, d_8 > 0, \\
    \text{and} \, \max \lcurl{z_7^-, z_8^-, \alpha_5} < \min \lcurl{z_7^+, z_8^+, \alpha_2}.
  \end{gathered}
\end{equation}

\myformsec{formsec:r5_2a_1b} $\lpar{\C_{1 \leq 2}^2 \to \C_{1 > 2}^1}$: The analysis is identical to \textbf{\ref{formsec:r4_2a_1b}}.

\myformsec{formsec:r5_3a_3b} $\lpar{\C_{1 \leq 2}^3 \to \C_{1 > 2}^2}$: In order for $\OG \in \C_{1 > 2}^3$, $\nu$ must satisfy \eqref{eq:cross_middle_line} and \eqref{eq:stay_in_3}. The pre-transfer payoffs $U_1$ and $U_2$ are given by \eqref{eq:p3a} and \eqref{eq:p3b} with $i = 1$ and $\nu = 0$, and the post-transfer payoffs are given by \eqref{eq:p3a} and \eqref{eq:p3b} with $i = 1$. The condition $U_1(\tau, 0 \suchthat G) < U_1(\tau, \nu \suchthat G)$ can be written as
\begin{equation*}
    \nu < \frac{X_1}{X_1 + X_2} \lpar{\pq - \pw - 2 \mysqrt{\frac{X_1 \pq \pw}{X_2}}},
\end{equation*}
but it is readily verified that this condition and \eqref{eq:cross_middle_line} cannot be satisfied simultaneously. Thus, there does not exist any possible mutually beneficial transfer of this form.

\myformsec{formsec:r5_3a_2b} $\lpar{\C_{1 \leq 2}^3 \to \C_{1 > 2}^2}$: In order for $\OG \in \C_{1 > 2}^2$, $\nu$ must satisfy \eqref{eq:go_to_2b_2} and \eqref{eq:stay_in_2b}. The pre-transfer payoffs $U_1$ and $U_2$ are given by \eqref{eq:p3a} and \eqref{eq:p3b} with $i = 1$ and $\nu = 0$, and the post-transfer payoffs are given by \eqref{eq:p2b} and \eqref{eq:p2a} with $i = 2$. The conditions \eqref{eq:mut_ben_transfer} can be written in the form \eqref{eq:quadratic_inequality} with
\begin{gather*}
    a_{13} = 1 + \frac{X_1}{X_2} \lpar{2 - \frac{1}{X_1}}^2, \\
    b_{13} = \frac{4 X_1 - 2}{X_2} \lpar{\frac{(X_1 - 1)^2}{X_1} \pq + \mysqrt{\pq \pw X_1 X_2}} \\ 
    + \pw - \pq, \\
    c_{13} = \frac{X_1}{X_2} \lpar{\frac{(X_1 - 1)^2}{X_1} \pq + \mysqrt{\pq \pw X_1 X_2}}^2 - \pq \pw \\
    \text{and} \;
    a_{14} = 1, \; b_{14} = \pw - \pq, \\
    c_{14} = \frac{X_1}{X_2} \lpar{X_2 \pw + \mysqrt{X_1 X_2 \pq \pw} - \pq \pw}
\end{gather*}
for $U_1$ and $U_2$, respectively. Thus, there exists a positive transfer satisfying \eqref{eq:mut_ben_transfer}, \eqref{eq:go_to_2b_2}, and \eqref{eq:stay_in_2b} if and only if
\begin{equation}\label{eq:r5_c3a_c2b_cond}
  \begin{gathered}
    d_{13} > 0, d_{14} > 0, \\
    \text{and} \, \max \lcurl{z_{13}^-, z_{14}^-, \alpha_5} < \min \lcurl{z_{13}^+, z_{14}^+, \alpha_2}.
  \end{gathered}
\end{equation}

\myformsec{formsec:r5_3a_1b} $\lpar{\C_{1 \leq 2}^3 \to \C_{1 > 2}^1}$: In order for $\OG \in \C_{1 > 2}^1$, $\nu$ must satisfy \eqref{eq:go_to_1b}. The pre-transfer payoffs $U_1$ and $U_2$ are given by \eqref{eq:p3a} and \eqref{eq:p3b} with $i = 1$ and $\nu = 0$, and the post-transfer payoffs are given by \eqref{eq:p1b} and \eqref{eq:p1a} with $i = 2$. The conditions \eqref{eq:mut_ben_transfer} can be written as
\begin{equation*}
     \mysqrt{\frac{X_1 \pq \pw}{X_2}} < \pq \lpar{1 - \frac{X_1}{2}} - \frac{1}{2} \mysqrt{X_1 X_2 \pq \pw}.
\end{equation*}
Thus, there exists a positive transfer satisfying \eqref{eq:mut_ben_transfer} and \eqref{eq:go_to_1b} if and only if
\begin{equation}\label{eq:r5_c3a_c1b_cond}
  \begin{gathered}
    \max \lcurl{\mysqrt{\frac{X_1 \pq \pw}{X_2}}, \alpha_2} \\ < \pq \lpar{1 - \frac{X_1}{2}} - \frac{1}{2} \mysqrt{X_1 X_2 \pq \pw}.
  \end{gathered}
\end{equation}

Thus, we define
\begin{equation}
\begin{split}
    \G_{1 \leq 2}^{\rm SI, 5} = \big\{G \in \C_{1 \leq 2}^1 \, \cap \, \RR^5 \ \big\vert \ &\text{$G$ satisfies \eqref{eq:r5_c1a_c2a_cond}, \eqref{eq:r5_c1a_c3a_cond},}\\
    &\text{
    \eqref{eq:r5_c1a_c3b_cond}, 
    \eqref{eq:r5_c1a_c2b_cond}, or 
    \eqref{eq:r4_c1a_c1b_cond}}
    \big\} \\
    \cup \, \big\{G \in \C_{1 \leq 2}^2 \, \cap \, \RR^5 \ \big\vert \ &\text{$G$ satisfies \eqref{eq:r5_c2a_c3a_cond},
    \eqref{eq:r5_c2a_c3b_cond},} \\
    &\text{
    \eqref{eq:r5_c2a_c2b_cond}, or 
    \eqref{eq:r4_c2a_c1b_cond}}
    \big\} \\
    \cup \, \big\{G \in \C_{1 \leq 2}^2 \, \cap \, \RR^5 \ \big\vert \ &\text{$G$ satisfies \eqref{eq:r5_c3a_c2b_cond} or 
    \eqref{eq:r5_c3a_c1b_cond}}
    \big\}
\end{split}
\end{equation}
as the set of games in Region 5 in which Players can perform positive, strategically inconsistent, and mutually beneficial contest transfers. Collecting these subsets, we define
\begin{equation}
    \G_{1 \leq 2}^{\rm SI} = \G_{1 \leq 2}^{\rm SI, 1} \, \cup \, \G_{1 \leq 2}^{\rm SI, 2} \, \cup \, \G_{1 \leq 2}^{\rm SI, 3} \, \cup \, \G_{1 \leq 2}^{\rm SI, 4} \, \cup \, \G_{1 \leq 2}^{\rm SI, 5}
\end{equation}
as the set of game instances in which Players can perform positive, strategically inconsistent, and mutually beneficial contest transfers. We define
\begin{equation}
    \G_{1 \leq 2}^\nu = \G_{1 \leq 2}^{\rm SC} \, \cup \, \G_{1 \leq 2}^{\rm SI}
\end{equation}
as the set of games in $\G_{1 \leq 2}$ in which Players can perform positive, mutually beneficial contest transfers, and we define
\begin{equation}\label{eq:region_mut_ben_exist}
    \G^\nu = \G_{1 \leq 2}' \, \cup \, \G_{1 > 2}^\nu
\end{equation}
as the set of games where Players can perform mutually beneficial contest transfers, where $\G_{1 > 2}^\nu \subset (\G \setminus \G_{1 \leq 2})$ is defined similarly with the appropriate swapping of indices.

\subsubsection{Establishing the statements}

Statement a) follows from the characterization of $\G^\nu$ above, and from \cite{kovenock2012coalitional} and \cite{shah2024battlefield}. Statement b) follows by means of an example; it is readily verified that the game $G = (12, 10, 0.4, 1.6)$ belongs to both $\G^\tau$ and $\G^\nu$. Statement c) can similarly be verified by the construction of two examples. Statement d) is proven in \cite{shah2024battlefield}.

\subsection{Proof of Theorem \ref{thm:col_ben_equiv}}\label{app:equivalence_proof}

The proof requires analyzing the change in the collective payoff $U_{1 + 2}(\tau, \nu) \triangleq U_1(\tau, \nu) + U_2(\tau, \nu)$ as a function of the contest transfer amount $\nu$ for fixed $\tau = 0$ in each of the four Cases; the comparison with budget and joint transfers is shown at the end of the proof. A transfer from Player 1 to Player 2 effectively translates the game $G$ in the parameter space to a new game $\OG$, meaning that the game $\OG$ may belong to a Case different from that of $G$; this is detailed in Figure \ref{fig:regions}.
    
First, consider any game $G \in \C_{1 \leq 2}^1$. In this case, the Adversary allocates all of their budget to General Lotto game 1, meaning that the alliance wins a portion of the contests in Lotto game 1, and they win all the contests in Lotto game 2. Thus, it is clear that the alliance can improve their collective  payoff by transferring contest valuation from Player 1 to Player 2, increasing the value of the contests won in the game 2 (i.e., they win $\pw + \nu$); conversely, it is also clear that the alliance cannot improve their payoff by transferring valuation from Player 2 to Player 1. The alliance can improve their payoff by transferring contests from Player 1 to Player 2 until either $\OG \in \C_{1 \leq 2}^2$, which we address next, or $\OG \in \C^4$.

Second, consider any game $G \in \C_{1 \leq 2}^2$. In this case, the alliance payoff $ U_{1 + 2}(0, \nu \suchthat G)$ is given by
\begin{equation*}
    \lpar{1 - \frac{1}{2 X_2}} (\pw + \nu) + \mysqrt{\frac{X_1 (\pq - \nu) (\pw + \nu)}{X_2}},
\end{equation*}
and its derivative $\dnu{U_{1 + 2}}$ is given by
\begin{equation*}
    1 - \frac{1}{2 X_2} + \frac{1}{2} \mysqrt{\frac{X_1}{X_2 (\pq - \nu) (\pw + \nu)}} (\pq - \pw - 2 \nu).
\end{equation*}
It is readily shown that $\dnu{U_{1 + 2}}|_{\nu = 0} > 0$ when $X_2 > \frac{\pw}{\pq + \pw}$, which is true for all $G \in \C_{1 \leq 2}^2$. Hence, we have that collectively beneficial transfers exist for every game in $\C_{1 \leq 2}^2$. Furthermore, since $\dnu{U_{1 + 2}}|_{\nu = 0} > 0$ and $U_{1 + 2}$ is concave in $\nu$, we also have that collectively beneficial transfers must be positive. The alliance can thus improve its payoff by performing a transfer, which yields $\OG \in \C_{1 \leq 2}^3$ if $X_1 + X_2 < 1$, or $\OG \in \C^4$ otherwise—we address both Cases next.

Third, consider any game $G \in \C_{1 \leq 2}^3$. In this case, the alliance payoff $U_{1 + 2}(0, \nu \suchthat G)$ is given by
\begin{equation*}
    \frac{1}{2} X_1 (\pq - \nu) + \frac{1}{2} X_2 (\pw + \nu) + \mysqrt{X_1 X_2 (\pq  - \nu)(\pw + \nu)}
\end{equation*}
and its derivative $\dnu{U_{1 + 2}}$ is given by
\begin{equation*}
    \frac{1}{2} (X_2 - X_1) + \frac{1}{2} \mysqrt{\frac{X_1 X_2}{(\pq - \nu) (\pw + \nu)}} (\pq - \pw - 2 \nu).
\end{equation*}
It is straightforward to show that $\dnu{U_{1 + 2}}|_{\nu = 0}$ is nonnegative for all games in $\C_{1 \leq 2}^3$. Furthermore, simple algebraic substitution shows that the unique transfer $\nu^{**}$ that solves $\dnu{U_{1 + 2}} = 0$ satisfies $\frac{X_1}{\pq - \nu^{**}} = \frac{X_2}{\pw + \nu^{**}}$. Thus, collectively beneficial transfers exist for every game in $\C_{1 \leq 2}^3$, and they must go from Player 1 to Player 2.

Last, consider any game $G \in \C^4$. The previous arguments concerning Cases 1 and 2 establish that the alliance payoff increases only via positive transfers approaching Case 4. Furthermore, these arguments are symmetric for games $G \in \G \setminus \G_{1 \leq 2}$. This, along with the fact that $U_{1 + 2}$ is continuous in $\nu$, implies that the alliance payoff is maximized when $G \in \C^4$. By extension, for any game $G \in \C_{1 \leq 2}^1 \, \cup \, \C_{1 \leq 2}^2$, the transfer $\nu^{**}$ that maximizes the alliance payoff must be positive and must induce $\OG \in \C^4$, meaning that $\frac{X_1}{\pq - \nu^{**}} = \frac{X_2}{\pw + \nu^{**}}$.

The previous arguments establish that the collectively beneficial transfer that maximizes the collective payoff is such that $\frac{X_1}{\ophione} = \frac{X_2}{\ophitwo}$. To complete the proof, we must simply compare the payoff resulting from maximizing budget, contest, and joint transfers. First, observe that if $X_1 + X_2 \geq 1$, then from the above analysis and from \cite{kovenock2012coalitional}, the maximizing budget or contest transfer is such that $\overline{G} \in \C^4$, meaning that the collective payoff is
\begin{equation*}
    \frac{2(\overline{X}_1 + \overline{X_2}) - 1}{2(\overline{X}_1 + \overline{X_2})} \lpar{\ophione + \ophitwo} = \frac{2(X_1 + X_2) - 1}{2(X_1 + X_2)} \lpar{\pq + \pw},
\end{equation*}
which is a fixed constant regardless of the transfer type. Otherwise, if $X_1 + X_2 < 1$, then $\OG \in \C_3$, and the collective payoff is
\begin{gather*}
    \frac{1}{2} \lpar{\overline{X}_1 \ophione + \overline{X}_2 \ophitwo} + \mysqrt{\overline{X}_1 \overline{X}_2 \ophione \ophitwo}.
\end{gather*}
When $\frac{\overline{X}_1}{\ophione} = \frac{\overline{X}_2}{\ophitwo}$ (which is true for any maximizing contest or budget transfer \cite{kovenock2012coalitional}), this expression can be rewritten as $\frac{1}{2} \lpar{\pq + \pw} \lpar{X_1 + X_2}$, which is also a constant.

It is readily verified that this reasoning extends to joint transfers, since joint transfers are a pair of contest and budget transfers; maximizing joint transfers must also equate the Players' relative strengths, so the maximizing payoffs are also identical. Finally, for games $\G \setminus \G_{1 \leq 2}$, the same arguments can be carried out with only a swapping of indices.

\section*{References}
\bibliographystyle{ieeetr}
\bibliography{references}

\begin{thebibliography}{10}

\bibitem{elmuti2001overview}
D.~Elmuti and Y.~Kathawala, ``An overview of strategic alliances,'' {\em Management Decision}, vol.~39, no.~3, pp.~205--218, 2001.

\bibitem{culpan2002global}
R.~Culpan, {\em {Global Business Alliances: Theory and Practice}}.
\newblock Bloomsbury Publishing USA, 2002.

\bibitem{king2003complementary}
D.~R. King, J.~G. Covin, and W.~H. Hegarty, ``Complementary {Resources} and the {Exploitation} of {Technological} {Innovations},'' {\em Journal of Management}, vol.~29, no.~4, pp.~589--606, 2003.

\bibitem{van2015power}
T.~Van Der~Schoor and B.~Scholtens, ``{Power to the people: Local community initiatives and the transition to sustainable energy},'' {\em Renewable and sustainable energy reviews}, vol.~43, pp.~666--675, 2015.

\bibitem{de2020cooperatives}
M.~de~Bakker, A.~Lagendijk, and M.~Wiering, ``Cooperatives, incumbency, or market hybridity: {New} alliances in the {Dutch} energy provision,'' {\em Energy Research \& Social Science}, vol.~61, p.~101345, 2020.

\bibitem{tushar2018transforming}
W.~Tushar, C.~Yuen, H.~Mohsenian-Rad, T.~Saha, H.~V. Poor, and K.~L. Wood, ``Transforming energy networks via peer-to-peer energy trading: The potential of game-theoretic approaches,'' {\em IEEE Signal Processing Magazine}, vol.~35, no.~4, pp.~90--111, 2018.

\bibitem{blais2007making}
A.~Blais and I.~H. Indridason, ``{Making Candidates Count: The Logic of Electoral Alliances in Two-Round Legislative Elections},'' {\em The Journal of Politics}, vol.~69, no.~1, pp.~193--205, 2007.

\bibitem{di1998electoral}
A.~Di~Virgilio, ``{Electoral alliances: Party identities and coalition games},'' {\em European Journal of Political Research}, vol.~34, no.~1, pp.~5--33, 1998.

\bibitem{spoon2015alone}
J.-J. Spoon and K.~J. West, ``{Alone or together? How institutions affect party entry in presidential elections in Europe and South America},'' {\em Party Politics}, vol.~21, no.~3, pp.~393--403, 2015.

\bibitem{bitar2012}
E.~Y. Bitar, E.~Baeyens, P.~P. Khargonekar, K.~Poolla, and P.~Varaiya, ``Optimal sharing of quantity risk for a coalition of wind power producers facing nodal prices,'' in {\em 2012 American Control Conference (ACC)}, pp.~4438--4445, 2012.

\bibitem{kaewpuang2013framework}
R.~Kaewpuang, D.~Niyato, P.~Wang, and E.~Hossain, ``A framework for cooperative resource management in mobile cloud computing,'' {\em IEEE Journal on Selected Areas in Communications}, vol.~31, no.~12, pp.~2685--2700, 2013.

\bibitem{zhang2019competition}
B.~Zhang, R.~Johari, and R.~Rajagopal, ``Competition and efficiency of coalitions in cournot games with uncertainty,'' {\em IEEE Transactions on Control of Network Systems}, vol.~6, no.~2, pp.~884--896, 2019.

\bibitem{zhang2020}
Z.~Zhang, Y.~Jiang, Z.~Lin, F.~Wen, Y.~Ding, L.~Yang, Z.~Lin, Y.~Li, H.~Qian, J.~Li, and C.~He, ``Optimal alliance strategies among retailers under energy deviation settlement mechanism in china's forward electricity market,'' {\em IEEE Transactions on Power Systems}, vol.~35, no.~3, pp.~2059--2071, 2020.

\bibitem{saad2009coalitional}
W.~Saad, Z.~Han, M.~Debbah, A.~Hjorungnes, and T.~Basar, ``Coalitional game theory for communication networks,'' {\em IEEE Signal Processing Magazine}, vol.~26, no.~5, pp.~77--97, 2009.

\bibitem{bou2013}
E.~Bou-Harb, C.~Fachkha, M.~Pourzandi, M.~Debbabi, and C.~Assi, ``Communication security for smart grid distribution networks,'' {\em IEEE Communications Magazine}, vol.~51, no.~1, pp.~42--49, 2013.

\bibitem{kovenock2012coalitional}
D.~Kovenock and B.~Roberson, ``{Coalitional Colonel Blotto Games with Application to the Economics of Alliances},'' {\em Journal of Public Economic Theory}, vol.~14, no.~4, pp.~653--676, 2012.

\bibitem{shah2024battlefield}
V.~Shah and J.~R. Marden, ``{Battlefield Transfers in Coalitional Blotto Games},'' in {\em Proceedings of the 23rd International Conference on Autonomous Agents and Multiagent Systems}, pp.~1710--1717, 2024.

\bibitem{shah2024inefficient}
V.~Shah, K.~Paarporn, and J.~R. Marden, ``{Inefficient Alliance Formation in Coalitional Blotto Games},'' {\em IEEE Control Systems Letters}, 2024.

\bibitem{chandan2020showing}
R.~Chandan, K.~Paarporn, and J.~R. Marden, ``{When showing your hand pays off: Announcing strategic intentions in Colonel Blotto games},'' in {\em 2020 American Control Conference (ACC)}, pp.~4632--4637, IEEE, 2020.

\bibitem{chandan2022art}
R.~Chandan, K.~Paarporn, D.~Kovenock, M.~Alizadeh, and J.~R. Marden, ``{The Art of Concession in General Lotto Games},'' in {\em International Conference on Game Theory for Networks}, pp.~310--327, Springer, 2022.

\bibitem{gupta2014three}
A.~Gupta, G.~Schwartz, C.~Langbort, S.~S. Sastry, and T.~Ba{\v{r}}ar, ``{A three-stage Colonel Blotto game with applications to cyberphysical security},'' in {\em 2014 American Control Conference}, pp.~3820--3825, IEEE, 2014.

\bibitem{gupta2014three2}
A.~Gupta, T.~Ba{\c{s}}ar, and G.~A. Schwartz, ``{A three-stage Colonel Blotto game: when to provide more information to an adversary},'' in {\em Decision and Game Theory for Security: 5th International Conference, GameSec 2014, Los Angeles, CA, USA, November 6-7, 2014. Proceedings 5}, pp.~216--233, Springer, 2014.

\bibitem{heyman2018colonel}
J.~L. Heyman and A.~Gupta, ``{Colonel Blotto Game with Coalition Formation for Sharing Resources},'' in {\em Decision and Game Theory for Security} (L.~Bushnell, R.~Poovendran, and T.~Ba{\c{s}}ar, eds.), (Cham), pp.~166--185, Springer International Publishing, 2018.

\bibitem{diaz2023beyond}
G.~Díaz-Garcia, F.~Bullo, and J.~R. Marden, ``{Beyond the ‘Enemy-of-my-Enemy’ Alliances: Coalitions in Networked Contest Games},'' in {\em 2023 62nd IEEE Conference on Decision and Control (CDC)}, pp.~2220--2225, 2023.

\bibitem{diaz2024strategic}
G.~Diaz-Garcia, F.~Bullo, and J.~R. Marden, ``{Strategic Coalitions in Networked Contest Games},'' {\em arXiv preprint arXiv:2408.00883}, 2024.

\bibitem{paarporn2021division}
K.~Paarporn, R.~Chandan, M.~Alizadeh, and J.~R. Marden, ``The division of assets in multiagent systems: A case study in team blotto games,'' in {\em 2021 60th IEEE Conference on Decision and Control (CDC)}, pp.~1663--1668, IEEE, 2021.

\bibitem{igami2022measuring}
M.~Igami and T.~Sugaya, ``{Measuring the Incentive to Collude: The Vitamin Cartels, 1990–99},'' {\em The Review of Economic Studies}, vol.~89, no.~3, pp.~1460--1494, 2022.

\bibitem{kovenock2021generalizations}
D.~Kovenock and B.~Roberson, ``{Generalizations of the General Lotto and Colonel Blotto games},'' {\em Economic Theory}, vol.~71, pp.~997--1032, 2021.

\end{thebibliography}

\begin{IEEEbiography}[{\includegraphics[width=1in,height=1.25in,clip,keepaspectratio]{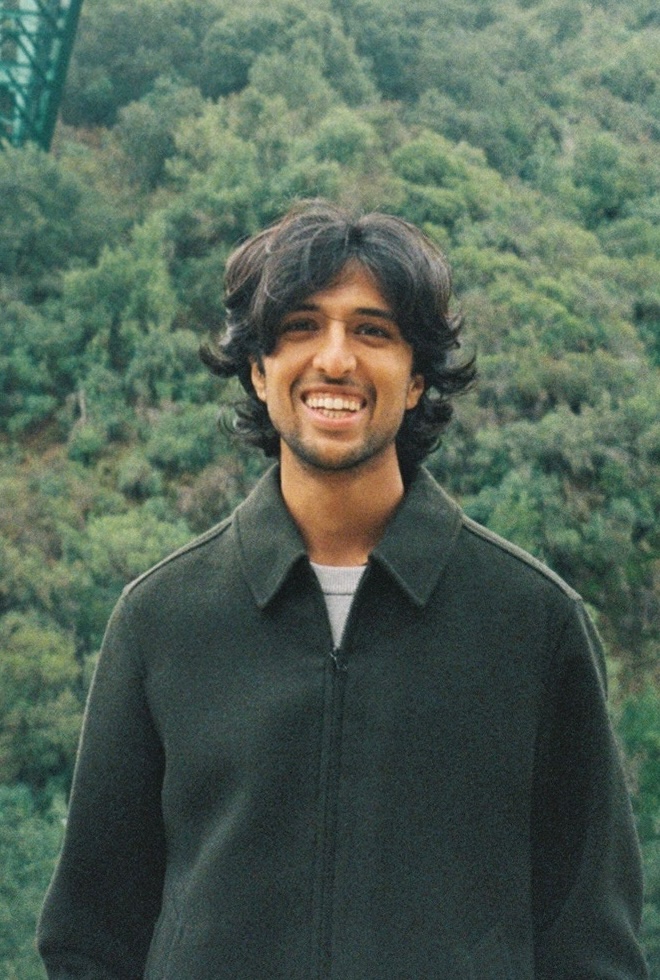}}]{Vade Shah} (Graduate Student Member, IEEE) received a B.S. in Electrical Engineering and Computer Sciences from UC Berkeley in 2022 and an M.S. in Electrical Engineering from UC Santa Barbara in 2024, where he is currently pursuing a Ph.D. under the supervision of Jason R. Marden. Vade is a recipient of the NSF GRFP (2024) and the SPIE Alexander Glass Best Oral Presentation Award (2024), and he was named a MS\&E Rising Star in 2024. His research interests focus on applying game theory to engineered systems.
\end{IEEEbiography}

\begin{IEEEbiography}[{\includegraphics[width=1in,height=1.25in,clip,keepaspectratio]{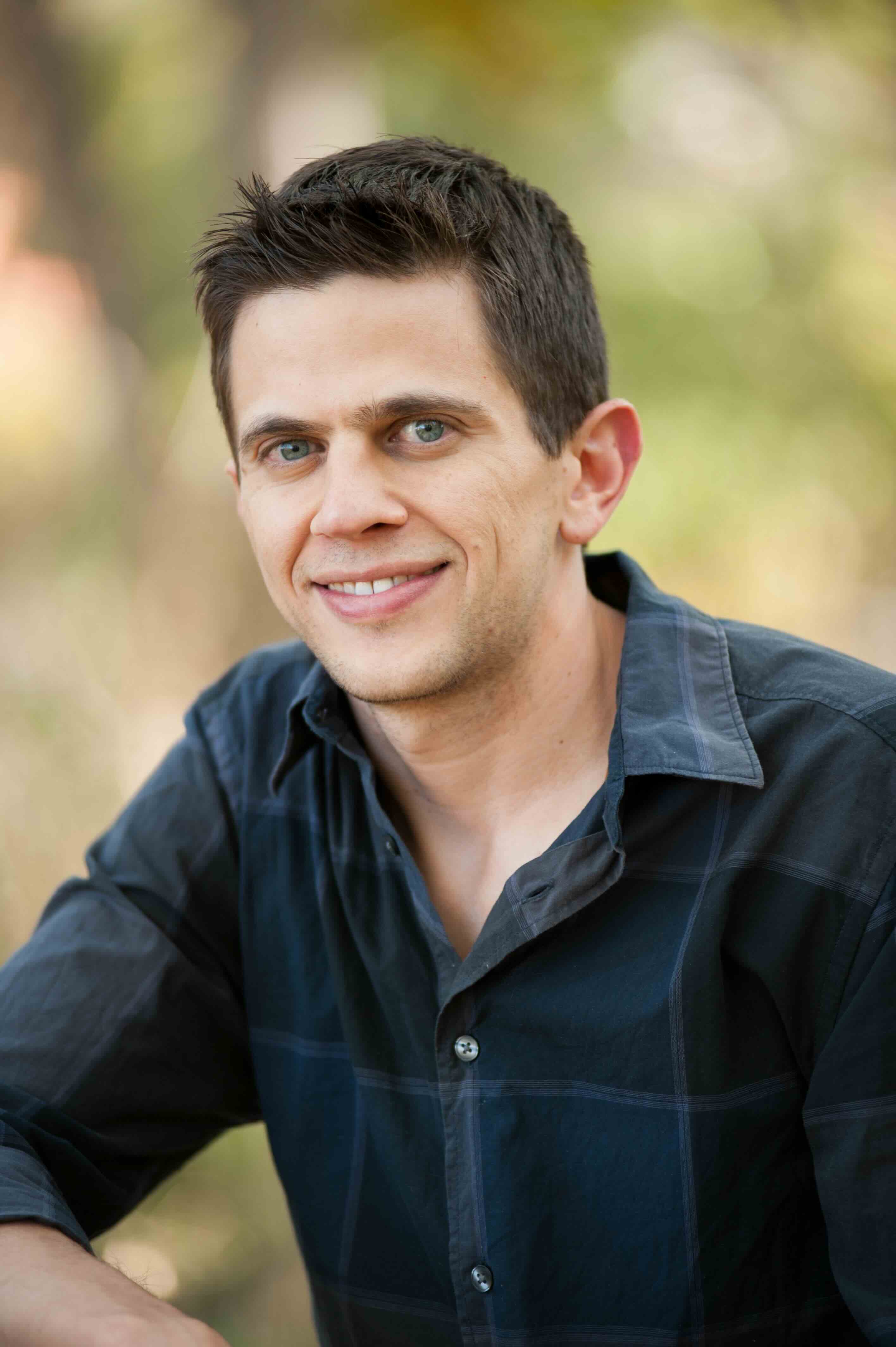}}]{Jason R. Marden} (Fellow, IEEE) is a Professor in the Department of Electrical and Computer Engineering at the University of California, Santa Barbara. Jason received a BS in Mechanical Engineering in 2001 from UCLA, and a PhD in Mechanical Engineering in 2007, also from UCLA, under the supervision of Jeff S. Shamma, where he was awarded the Outstanding Graduating PhD Student in Mechanical Engineering. After graduating from UCLA, he served as a junior fellow in the Social and Information Sciences Laboratory at the California Institute of Technology until 2010 when he joined the University of Colorado. Jason is a recipient of the NSF Career Award (2014), the ONR Young Investigator Award (2015), the AFOSR Young Investigator Award (2012), the American Automatic Control Council Donald P. Eckman Award (2012), and the SIAG/CST Best SICON Paper Prize (2015). Jason’s research interests focus on game theoretic methods for the control of distributed multiagent systems.
\end{IEEEbiography}

\end{document}